\documentclass[12pt]{article}

\usepackage{amsmath}
\usepackage{amssymb}
\usepackage{amsthm}
\usepackage{graphicx}
\newcommand{\be}{\begin{equation}}
\newcommand{\ee}{\end{equation}}
\newcommand{\ba}{\begin{array}}
\newcommand{\ea}{\end{array}}
\newcommand{\bea}{\begin{eqnarray}}
\newcommand{\eea}{\end{eqnarray}}

\newcommand{\calH}{{\cal H }}
\newcommand{\calL}{{\cal L }}

\newcommand{\calE}{{\cal E }}

\newcommand{\calP}{{\cal P }}
\newcommand{\calU}{{\cal U }}

\newcommand{\calO}{{\cal O }}
\newcommand{\calD}{{\cal D }}

\newcommand{\calT}{{\cal T }}
\newcommand{\calI}{{\cal I }}

\newcommand{\CC}{\mathbb{C}}
\newcommand{\RR}{\mathbb{R}}
\newcommand{\LL}{\mathbb{L}}
\newcommand{\la}{\langle}
\newcommand{\ra}{\rangle}

\newcommand{\nn}{\nonumber}

\newcommand{\trace}{\mathrm{Tr} \,}

\newcommand{\gl}{\mathrm{gl}}
\newcommand{\loc}{\mathrm{loc}}
\newcommand{\eff}{\mathrm{eff}}
\newcommand{\Vd}{V_{\mathrm{d}}}
\newcommand{\Vod}{V_{\mathrm{od}}}

\newcommand{\els}{\mathrm{else}}
\newcommand{\garb}{\mathrm{garbage}}

\newtheorem{dfn}{Definition}
\newtheorem{lemma}{Lemma}
\newtheorem{prop}{Proposition}
\newtheorem{theorem}{Theorem}
\newtheorem{corol}{Corollary}

\numberwithin{lemma}{section}
\numberwithin{corol}{section}
\numberwithin{prop}{section}
\numberwithin{dfn}{section}
\numberwithin{equation}{section}

\begin{document}

\title{Schrieffer-Wolff transformation for quantum many-body systems}
\author{Sergey Bravyi\thanks{IBM Watson
Research Center, Yorktown Heights, NY 10598 USA.
\texttt{sbravyi@us.ibm.com} }, \, \,
David P. DiVincenzo\thanks{RWTH Aachen and Forschungszentrum Juelich, Germany.
\texttt{d.divincenzo@fz-juelich.de} }, \, \, and \, \,
Daniel Loss\thanks{Department of Physics, University of Basel,
Klingelbergstrasse 82, CH-4056 Basel, Switzerland.
\texttt{Daniel.Loss@unibas.ch}}}

\maketitle

\begin{abstract}
The Schrieffer-Wolff (SW) method is a version of degenerate perturbation theory in which the  low-energy effective Hamiltonian $H_{\eff}$ is obtained from the exact Hamiltonian by a unitary transformation decoupling the low-energy and high-energy subspaces. We give a self-contained summary of the SW method with a focus on rigorous results.  We begin with an exact definition of the SW transformation in terms of the so-called direct rotation between linear subspaces. From this we obtain elementary proofs of several important properties of $H_{\eff}$ such as the linked cluster theorem.  We then study the perturbative version of the SW transformation obtained from a Taylor series representation of the direct rotation. Our perturbative approach provides a systematic diagram technique for computing high-order corrections to $H_{\eff}$.  We then specialize the SW method to quantum spin lattices with short-range interactions. We establish unitary equivalence between effective low-energy Hamiltonians obtained using two different versions of  the SW method studied in the literature.  Finally, we derive an upper bound on the precision up to which the ground state energy of the $n$-th order effective Hamiltonian approximates the exact ground state energy.
%
 \end{abstract}

\newpage

\tableofcontents

\newpage

\section{Introduction}
Given a fundamental theory describing a quantum many-body system, one often needs to
obtain a concise description of its low-energy dynamics by integrating out
high-energy degrees of freedom. The Schrieffer-Wolff  (SW) method accomplishes this task  by
constructing a unitary transformation that decouples the
high-energy and low-energy subspaces.

In the present paper we view the SW method as a version of degenerate perturbation theory. It
involves a Hamiltonian $H_0$ describing an unperturbed system, a low-energy subspace $\calP_0$
invariant under $H_0$,
and a perturbation $\epsilon V$ which does not preserve $\calP_0$. The goal is to construct an effective Hamiltonian
$H_{\eff}$ acting only on $\calP_0$ such that the spectrum of $H_{\eff}$ reproduces eigenvalues of
the perturbed Hamiltonian $H_0+\epsilon V$ originating from the low-energy subspace $\calP_0$.
The SW transformation is
a unitary operator $U$ such that the transformed Hamiltonian $U (H_0+\epsilon V)U^\dag$ preserves $\calP_0$.
The desired effective Hamiltonian is then defined as the restriction of $U (H_0+\epsilon V)U^\dag$
onto $\calP_0$.
An important  consideration arising in the context of many-body systems such as molecules or interacting
quantum spins is the locality of $H_{\eff}$.  In order for the effective Hamiltonian $H_{\eff}$ to be usable, it must
only involve $k$-body interactions for some small constant $k$. We demonstrate
that this is indeed the case for the SW method by proving (under certain natural assumptions)
that it obeys the so called linked cluster theorem.

The purpose of the present paper is two-fold. First, we give a self-contained summary of the SW method with a focus on rigorous results. A distinct feature of our presentation is the use of both perturbative and exact treatments of the SW transformation. The former is given in terms of the Taylor series and the rules for computing the Taylor coefficients while
the latter involves the so called direct rotation between linear subspaces\footnote{The direct rotation
can be thought of as a square root of the double reflection operator used in the Grover search algorithm~\cite{Grover}.}.
By combining these different perspectives
we are able to obtain elementary proofs of several important properties of the effective Hamiltonian $H_{\eff}$ such as its additivity under a disjoint union of non-interacting systems.
These properties  do not manifest themselves on the level of individual terms in the perturbative expansion which makes their direct proof (by inspection of the Taylor coefficients of $H_{\eff}$)
virtually impossible. In addition, our approach provides a systematic  diagram technique for computing high-order corrections to $H_{\eff}$ which can be readily cast into a computer program for practical calculations.

Our second goal is to specialize the SW method to interacting quantum spin systems.
We  assume that the unperturbed Hamiltonian $H_0$ is a sum of single-spin
operators such that the low-energy subspace $\calP_0$
is a tensor product of single-spin  low-energy subspaces.
The perturbation $V$ is chosen as a sum of two-spin interactions associated with edges
of some fixed interaction graph.  We prove that the Taylor series for $H_{\eff}$
obeys the so called linked cluster theorem, that is, the $m$-th order correction to $H_{\eff}$
includes only interactions among subsets of spins spanned by
connected clusters of at most $m$ edges in the interaction graph (one should keep in mind that the spins acted on
by $H_{\eff}$ are `effective spins' described by the low-energy subspaces of the original spins),
see Theorem~\ref{thm:lct} in Section~\ref{subs:lct}.
It demonstrates that the SW transformation maps a high-energy theory with local interactions to
an effective low-energy theory with (approximately) local interactions.

For macroscopic spin systems the norm of the perturbation $\|\epsilon V\|$ is typically much larger than the
spectral gap of $H_0$ and one should expect level crossings to occur between energy levels from the low-energy and high-energy subspaces. On the other hand, one should expect that for sufficiently small $|\epsilon|$ ground states of $H_0+\epsilon V$
have most of their weight inside the low-energy subspace $\calP_0$.
We prove that this is indeed the case for so-called block-diagonal perturbations $V$, that is,
perturbations  preserving $\calP_0$.
Specifically, we prove that the low-energy subspace $\calP_0$ contains at least one ground state of $H_0+\epsilon V$
as long as $|\epsilon|<\epsilon_0$ for some constant  $\epsilon_0$ that depends only on the
spectral gap of $H_0$, the degree of the interaction graph, and the strength of spin-spin interactions,
see Lemma~\ref{lemma:new} in Section~\ref{subs:bd}. In that sense,  the ground state
of the system has a constant stability radius against block-diagonal perturbations.

The main technical result reported in the paper is a rigorous upper bound on the error $\delta_m$
up to which the ground state energy of $H_{\eff}$,
truncated at the $m_{th}$ order,
approximates the ground state energy of the exact Hamiltonian $H_0+\epsilon V$,
see Theorem~\ref{thm:global}
in Section~\ref{subs:equiv}.
Specifically, we prove that there exist
constants $c_m, \epsilon_m$ such that
$\delta_m \le c_m N |\epsilon|^{m+1}$ for all $|\epsilon|<\epsilon_m$,
where $N$ is the total number of spins.
 The  coefficients $c_m$ and $\epsilon_m$
depend only on the degree of the interaction graph, the strength of spin-spin interactions,
the spectral gap $\Delta$ of the unperturbed Hamiltonian $H_0$, and the truncation order $m$.
It shows that the approximation error $\delta_m$ is extensive (linear in $N$)
for fixed  $m$. On the other hand, $\delta_m$ may grow very rapidly as a function of $m$.
Note that our bound holds even in the region of parameters $\| \epsilon V\| \gg \Delta$
where the perturbative series for $H_{\eff}$ are expected to be divergent.  The proof of the bound
on $\delta_m$ combines the SW transformation, the linked cluster theorem, and
the ground state stability against block-diagonal perturbations.

One  undesirable feature of the SW transformation is that its generator $S=\log{(U)}$
is a highly non-local operator which cannot be written as a sum of local interactions
even in the lowest order in $\epsilon$. The fact that locality of the theory does not
manifest itself on the level of the generator $S$ can be safely ignored if one is interested only in the
final expression for  $H_{\eff}$. However the non-local character of $S$
makes it more difficult to prove bounds on the truncation error $\delta_m$.
We resolve this problem by employing a different version of the SW method
due to Datta et al~\cite{DFFR} which we refer to as a {\em local SW method}.
By analogy with the standard SW it defines
a unitary transformation $U$ generating an effective low-energy Hamiltonian acting
on $\calP_0$.  The advantage of the local SW transformation is that
its generator $S=\log{(U)}$ has nice  locality properties  which makes analysis of the
truncation error $\delta_m$ obtained using the local SW method much easier,
see Theorem~\ref{thm:local} in Section~\ref{subs:DFFR}.
Unfortunately, the local SW transformation does not exist outside the realm
of perturbative expansions lacking an `exact definition' similar to the standard SW.
We establish equivalence between the local and the standard versions
of the SW method by  showing that the corresponding $m$-th order effective Hamiltonians
can be mapped to each other by a unitary rotation of the low-energy subspace
$\calP_0$ up to an error $O(N\epsilon^{m+1})$.

\subsection{Applications of the SW method}

The SW transformation finds application in many contemporary problems in quantum physics, where an economical description of the low-energy dynamics of the system must be extracted from a full Hamiltonian.  While SW is named in honor of the authors of an important paper in condensed matter theory~\cite{SW66}, where the celebrated Kondo Hamiltonian was shown to be obtained by such a transformation from the equally celebrated Anderson Hamiltonian, it is in fact not the earliest application of the technique.  The concept of a most-economical rotation between subspaces has a several-centuries history in mathematics~\cite{DavisKahan69,PW94}. Its direct descendent is the cosine-sine (CS) decomposition~\cite{Bhatia}, which is very closely related to the definition of the exact SW transformation that we adopt here.   The earliest application of the canonical transformation in quantum physics is in fact at least 15 years earlier than~\cite{SW66}, in the famous transformation of Foldy and Wouthuysen~\cite{FW50}.  The authors of~\cite{FW50} used the technique to derive the non-relativistic Schr\"oedinger-Pauli wave equation, with relativistic corrections, from the Dirac equation.  In this case, the canonical transformation is one that decouples the positive and negative energy solutions of the relativistic wave equation.

The applications of SW in modern times~\cite{W86} are far too numerous even to allude to.  The work of one of us uses SW routinely in many applications (see, e.g., Ref.~\cite{GKL08}).
 Since part of our concern in the present work is a systematic development of a proper series expansion for the SW, we can mention that  high-order terms of this series have been carefully and laboriously developed in some important works~\cite{MGY88} and may be found recorded in recent books dealing with applications~\cite{W10}.  One can, however, find lengthy discussions of possible alternative approaches~\cite{O90,MGY90}.  It is widely acknowledged that SW is preferable to various other approaches, for example the Bloch expansion~\cite{Bloch58,Lindgren73}, which, unlike SW, produces non-Hermitian effective Hamiltonians at sufficiently high order.  SW is also superior to the self-consistent equation for the effective Hamiltonian arising from the diagrammatic self-energy technique~\cite{AGD,Fetter,KKR}.

It often occurs that the effective low-energy Hamiltonian has much higher level of complexity compared to the original
high-energy Hamiltonian.  This observation has lead to the idea of perturbation gadgets~\cite{KKR,Terhal08,JordanFarhi:08}.
Suppose one starts from a target Hamiltonian $H_{\mathrm{target}}$
chosen for some interesting ground state properties. The target Hamiltonian may possibly
contain many-body interactions and might not be very realistic.
Perturbation gadgets formalism allows one  to construct a simpler high-energy simulator Hamiltonian
with only two-body interactions whose low-energy properties (such as the ground state energy) approximate the ones
of  $H_{\mathrm{target}}$. For example, various perturbation gadgets have been constructed for target Hamiltonians with the topological quantum order~\cite{Kitaev03,Koenig10,Brell11}. The SW method provides a natural framework
for constructing and analyzing perturbation gadgets~\cite{gadget_prl}. In contrast to other methods,
SW does not require the unphysical scaling of parameters in the simulator Hamiltonian
required for  convergence of the perturbative series~\cite{gadget_prl}. This makes the SW method
suitable for analysis of experimental implementations of perturbation gadgets with cold atoms in
optical lattices,
see~\cite{optical1,optical2}.

\subsection{Comparison between SW and other perturbative expansions}
\label{subs:comparison}

In this section we briefly review some of the commonly used perturbative expansion techniques
and argue that none of them can serve as a fully functional replacement of the SW method.

We begin by emphasizing that the  effective Hamiltonian $H_{\eff}$ is only defined
up to a unitary rotation of the low-energy subspace $\calP_0$. Therefore one should expect that
different versions of a degenerate perturbation theory produce
different Taylor series for $H_{\eff}$. In particular, the error up to which the series for $H_{\eff}$ truncated  at some finite order
reproduce the exact low-energy spectrum may depend on the method of computing $H_{\eff}$.
Different methods may also vary in the complexity of rules for computing the Taylor coefficients.

The most convenient and commonly used perturbative method is the Feynman-Dyson diagram
technique~\cite{AGD,Fetter,Negele}.
It can  be used whenever the unpertubed ground state obeys Wick's theorem. Unfortunately,
the standard  derivation of the Feynman-Dyson expansion (see, e.g., Ref.~\cite{Fetter}) which relies
on the adiabatic switching of a perturbation can only be applied to non-degenerate ground states.
It was recently shown explicitly that  the Gell-Mann and Low theorem used in the derivation
of Feynman-Dyson series fails for degenerate ground states~\cite{Brouder08}.

Exact quasi-adiabatic continuation~\cite{hwen,tjo}
provides an alternative path to defining $H_{\eff}$ via a unitary transformation
starting from a full Hamiltonian. This method however requires a constant lower bound on the energy gap
of a perturbed Hamiltonian which is typically very hard to check.

Traditional textbook treatment of a degenerate perturbation theory~\cite{Kato49,Messiah61,ReedSimon}
is formulated in terms of  perturbed and unperturbed resolvents
\[
\tilde{G}(z)=(zI - H_0 - \epsilon V)^{-1} \quad \mbox{and} \quad G(z)=(zI-H_0)^{-1}.
\]
An effective low-energy Hamiltonian  can be obtained from $\tilde{G}(z)$ using the self-energy method~\cite{AGD,Fetter,KKR}.
To sketch the method, we shall adopt the standard notations by writing any operator $O$
as a block matrix
\be
\label{pmnote}
O=\left[ \ba{cc}
O_{-} & O_{-+} \\ O_{+-} & O_{+} \\
\ea
\right]
\ee
where the two blocks correspond to the low-energy and the high-energy subspaces
respectively.
One can define an effective low-energy resolvent
as $\tilde{G}_{-}(z)$, that is, the low-energy block of $\tilde{G}(z)$. Its importance comes from the
fact that for sufficiently small $\epsilon$
eigenvalues of $H_0+\epsilon V$ originating from the low-energy subspace
coincide with the poles of $\tilde{G}_{-}(z)$, see~\cite{KKR}.
The poles of $\tilde{G}_{-}(z)$ can be found  using the Dyson equation~\cite{Fetter}
\[
\tilde{G}_{-}(z)=G_{-}(z) + \tilde{G}_{-}(z) \Sigma(z) G_{-}(z),
\]
where
\[
\Sigma(z)=\epsilon V_{-} + \sum_{n=0}^\infty \epsilon^{2+n} \, V_{-+} (G_{+}(z) V_{+})^n G_{+}(z) V_{+-}
\]
is the so-called self-energy operator acting on the low-energy subspace.
The Dyson equation yields
\[
\tilde{G}_{-}(z)^{-1} = G_{-}(z)^{-1} - \Sigma(z)=zI - (H_0)_{-} - \Sigma(z).
\]
Hence $z$ is an eigenvalue of $H_0+\epsilon V$ originating from the low-energy subspace
(the pole of  $\tilde{G}_{-}(z)$) iff $z$ is an eigenvalue of
$(H_0)_{-} +\Sigma(z)$. In that sense one can regard $H_{\eff}(z)\equiv (H_0)_{-} +\Sigma(z)$ as
a $z$-dependent effective Hamiltonian on the low-energy subspace.
Although in many important cases the $z$-dependence of $H_{\eff}(z)$ can be
neglected~\cite{KKR}, there is no systematic way of getting rid of this dependence.
This is the main disadvantage of the self-energy formalism compared with the SW method.

Another commonly used version of the degenerate perturbation theory is
the expansion due to Bloch~\cite{Bloch58,Messiah61,Lindgren73}, see~\cite{JordanFarhi:08,Koenig10}
for some recent applications. The Bloch expansion is conceptually very similar  to the SW method.
It provides a systematic way of constructing an ``effective Hamiltonian" acting on the low-energy subspace,
but  this Hamiltonian  is generally not hermitian and as such it does not describe any physical theory.
To make this point more clear let us recap the main steps of the Bloch expansion, see for instance~\cite{Lindgren73,JordanFarhi:08}.
Let $\{ |\psi_j\ra \}_j$ be the eigenvectors $H_0+\epsilon V$
originating from the low-energy subspace and $\{\lambda_j \}_j$ be the
corresponding eigenvalues.
Decompose
\[
|\psi_j\ra = |\alpha_j\ra + |\alpha_j^\perp\ra,
\quad
 \mbox{where} \quad
  |\alpha_j\ra \in \calP_0 \quad \mbox{and} \quad
|\alpha_j^\perp\ra \in \calP_0^\perp.
\]
Define an operator $\calU$
such that $\calU\, |\alpha_j \ra =|\psi_j\ra$ and $\calU\, |\psi\ra=0$ for all $\psi \in \calP_0^\perp$.
The operator $\calU$ is analogous to the (inverse) SW transformation, although it is not unitary.
Using the notations Eq.~(\ref{pmnote}) one can write $\calU$ as
\[
\calU=\left[ \ba{cc}
I & 0 \\ \calU_{+-} & 0 \\
\ea
\right].
\]
It shows that $\calU$  is a (non-hermitian) projector onto the perturbed low-energy subspace.
Define the effective Hamiltonian acting on $\calP_0$ as
\[
H_{\eff}=P_0(H_0+\epsilon V) \calU,
\]
where $P_0$ is the projector onto $\calP_0$. It follows that
\[
H_{\eff}\, |\alpha_j\ra = P_0(H_0+\epsilon V)\, |\psi_j\ra = \lambda_j\, P_0\, |\psi_j\ra= \lambda_j\, |\alpha_j\ra.
\]
Thus one can regard $H_{\eff}$ as a counterpart of the effective Hamiltonian in the SW theory --- it acts only on the
unperturbed low-energy subspace and its spectrum reproduces eigenvalues
of $H_0+\epsilon V$ originating from the low-energy subspace.
However, since $|\alpha_j\ra$ do not constitute an orthonormal basis of $\calP_0$, the operator $H_{\eff}$ is not hermitian, and thus it cannot be regarded as a physical Hamiltonian.

The advantage of the Bloch expansion is that the perturbative series for $H_{\eff}$ has somewhat simpler structure
compared with the corresponding expansion  in the SW method.
Indeed, using the fact that $\calU$ is a projector
one easily gets
\be
\label{Bloch1}
[H_0+\epsilon V,\calU]\calU=0.
\ee
Using the identities $P_0 \calU= P_0$ and $\calU P_0 = \calU$
one can rewrite Eq.~(\ref{Bloch1}) as
\be
\label{Bloch2}
[H_0,\calU]=-\epsilon [V,\calU] \calU=0.
\ee
Expanding $\calU=\sum_{n=0}^\infty \calU_n \epsilon^n$ with $\calU_0=P_0$ and using Eq.~(\ref{Bloch2}) one
can express $\calU_n$ as a second-degree polynomial
in $\calU_1,\ldots,\calU_{n-1}$ which provides a recursive rule for computing $\calU_n$.
This leads to a simple diagram technique for computing the Taylor coefficients of $H_{\eff}$,
see~\cite{Lindgren73,JordanFarhi:08} for details.

\subsection{Organization of the paper}
To be self-contained, the paper repeats some well-known results such as  the low-order terms in the SW series.
However  we present these results in a very systematic and compact form that we believe might be useful for other workers.

Sections~\ref{sec:crot},\ref{sec:SW} provide an elementary introduction to the SW method for
finite-dimensional Hilbert spaces. Section~\ref{sec:crot} summarizes the definition and basic facts
concerning  the direct rotation between linear subspaces. The direct rotation was originally introduced
in the context of perturbation theory by Davis and Kahan in~\cite{DavisKahan69}.
Although we hardly discover any new properties of the direct rotation, some of the basic facts such as
the behavior of the direct rotation under tensor products (see Section~\ref{subs:mult}) might not be very well-known.

Section~\ref{sec:SW} defines the SW transformation as the direct rotation
from the low-energy subspace of a perturbed Hamiltonian $H_0+\epsilon V$ to the low-energy
subspace of $H_0$. This definition provides an elementary
proof of  additivity of $H_{\eff}$. Namely, given a bipartite system $AB$ with no interactions between $A$ and $B$,
the effective Hamiltonians $H_{\eff}[AB]$, $H_{\eff}[A]$, and $H_{\eff}[B]$ describing the joint system
and the individual systems respectively are related as
\[
H_{\eff}[AB] =H_{\eff}[A] \otimes I_B + I_A \otimes H_{\eff}[B].
\]
Here $I_A,I_B$ are the identity operators acting on the low-energy subspaces of $A$ and $B$.
We note that the above additivity does not hold on the entire Hilbert space
as one could naively expect.
In Section~\ref{subs:series} we  compute the Taylor series for the SW transformation and $H_{\eff}$.
Our presentation uses the formalism of superoperators to develop the series. It allows us to keep track of cancelations between different terms in a systematic way and obtain more compact expressions for the Taylor coefficients.
A convenient diagram technique for computing $H_{\eff}$ is presented in Section~\ref{subs:diagram}.
Convergence of the series is analyzed in Section~\ref{subs:convergence}.
To the best of our knowledge, the approach taken in Sections~\ref{subs:series},\ref{subs:diagram},\ref{subs:convergence}
is new. Our main technical contributions are presented in Section~\ref{sec:manybody}
that specializes the SW method to weakly interacting spin systems.

\newpage
\section{Direct rotation between a pair of subspaces}
\label{sec:crot}
The purpose of this section is to define a direct rotation between a pair of linear subspaces
and derive its basic properties such as a multiplicativity under tensor products.

\subsection{Rotation of one-dimensional subspaces}
\label{subs:1d}
Let $\calH$ be any finite-dimensional Hilbert space.
For any normalized state $\psi \in \calH$
define a reflection operator $R_\psi$ that flips the sign of $\psi$ and acts trivially on the orthogonal
complement of $\psi$, that is,
\be
\label{Refl}
R_\psi=I-2|\psi\ra\la \psi|.
\ee
Given a pair of non-orthogonal states
$\psi,\phi \in \calH$  we would like to define a canonical unitary operator $U_{\psi\to \phi}$
mapping $\psi$ to $\phi$ up to an overall phase. Let us first consider the double reflection
operator $R_\phi R_\psi$. It rotates the two-dimensional subspace spanned by $\psi$ and $\phi$ by
$2\theta$, where $\theta\in [0,\pi/2)$ is the angle between $\psi$ and $\phi$.
In addition, $R_\phi R_\psi$ acts as the identity in the orthogonal complement to $\psi$ and $\phi$.
Thus we can choose the desired unitary operator mapping from $\psi$ to $\phi$ as  $U_{\psi\to \phi}=\sqrt{R_\phi R_\psi}$
assuming that the square root is well-defined.
\begin{dfn}
Let $\psi,\phi \in \calH$ be any non-orthogonal states.
Define a direct rotation from $\psi$ to $\phi$ as a unitary operator
\be
\label{canon1d}
U_{\psi\to \phi}=\sqrt{R_\phi R_{\psi}}
\ee
Here $\sqrt{z}$ is defined on a complex plane with a branch cut along the negative real axis
such that $\sqrt{1}=1$.
\end{dfn}
\noindent
It is worth pointing out that any unitary operator is normal and thus the square root
$\sqrt{R_\phi R_{\psi}}$ is well defined provided that $R_\phi R_{\psi}$ has
no eigenvalues lying on the chosen branch cut of the $\sqrt{z}$ function.
The following lemma shows that $U_{\psi\to \phi}$ is well defined and performs the desired transformation.
\begin{lemma}
\label{lemma:1d}
 Let $\psi,\phi \in \calH$ be any non-orthogonal states. Then the double reflection
 operator $R_\phi R_{\psi}$ has no eigenvalues on the negative real axis.
Furthermore,  fix the relative phase of $\psi$ and $\phi$ such that
 $\la \psi|\phi\ra$ is real and positive. Then $U_{\psi\to \phi}\, |\psi\ra=|\phi\ra$.
 \end{lemma}
\begin{proof}
Since $U_{\psi\to \phi}$ acts as the identity on the subspace orthogonal to $\psi$ and $\phi$, it suffices
to consider the case $\calH=\CC^2$. Without loss of generality $|\psi\ra=|1\ra$
and $|\phi\ra=\sin{(\theta)}\, |0\ra + \cos{(\theta)}\, |1\ra$, where, by assumption, $0\le \theta<\pi/2$.
Using the definition Eq.~(\ref{Refl}) one gets
\[
R_\psi=\sigma^z \quad \mbox{and} \quad  R_\phi=\cos{(2\theta)}\, \sigma^z - \sin{(2\theta)}\, \sigma^x.
\]
It yields
\[
R_\phi R_\psi=\cos{(2\theta)}\, I + i \sin{(2\theta)} \, \sigma^y = \exp{(2i\theta \sigma^y)}.
\]
Since $0\le 2\theta<\pi$, no eigenvalue of $R_\phi R_\psi$ lies on the negative real axis
and thus
\[
U_{\psi\to \phi} =  \sqrt{R_\phi R_\psi} =  \exp{(i\theta \sigma^y)}
\]
is uniquely defined.  We get $U_{\psi\to \phi} \, |1\ra = \sin{(\theta)}\, |0\ra + \cos{(\theta)}\, |1\ra$,
that is, $U_{\psi\to \phi}\, |\psi\ra =|\phi\ra$.
\end{proof}
Later on we shall need the following property of the direct rotation.
\begin{corol}
\label{cor:1d}
Let $\psi,\phi \in \calH$ be any non-orthogonal states,
$P=|\psi\ra\la \psi|$ and $P_0=|\phi\ra\la \phi|$. The direct
rotation from $\psi$ to $\phi$ can be written as
$U_{\psi\to \phi}=\exp{(S)}$ where $S$ is an anti-hermitian operator
with the following properties:
\begin{itemize}
\item $PSP=P_0 S P_0=(I-P)S(I-P)=(I-P_0)S(I-P_0)=0$,
\item $\|S\|<\pi/2$.
\end{itemize}
\end{corol}
\begin{proof}
Indeed, we can assume that $|\psi\ra=|1\ra$
and $|\phi\ra=\sin{(\theta)}\, |0\ra + \cos{(\theta)}\, |1\ra$ for some $\theta\in [0,\pi/2)$.
Choose $S=i\theta \sigma^y$ and  $S=0$ in the orthogonal complement
to $\psi$ and $\phi$. By assumption, $\|S\|=|\theta|<\pi/2$.
A simple algebra yields $\la \psi|S|\psi\ra=\la \phi|S|\phi\ra=0$.
Restricting all operators on the two-dimensional subspace spanned by $\psi$ and $\phi$
one gets $(I-P)S(I-P)=\la \psi^\perp|S|\psi^\perp\ra$ and $(I-P_0)S(I-P_0)=\la \phi^\perp|S|\phi^\perp\ra$,
where $|\psi^\perp\ra=|0\ra$ and $|\phi^\perp\ra=\cos{(\theta)}\, |0\ra-\sin{(\theta)}\, |1\ra$.
A simple algebra yields $\la \psi^\perp|S|\psi^\perp\ra=\la \phi^\perp|S|\phi^\perp\ra=0$.
\end{proof}
\subsection{Rotation of arbitrary subspaces}

Let $\calP,\calP_0 \subseteq \calH$ be a pair of linear subspaces of the same dimension
and $P$, $P_0$ be the orthogonal projectors onto $\calP,\calP_0$.
We would like to define a canonical unitary operator $U$ mapping $\calP$ to $\calP_0$.
By analogy with the non-orthogonality constraint used in the one-dimensional case we shall impose a
constraint
\be
\label{distance1}
\| P - P_0\|<1.
\ee
The meaning of this constraint is clarified by the following simple fact.
\begin{prop}
Condition $\| P - P_0\|<1$ holds iff no vector in $\calP$ is orthogonal to $\calP_0$
and vice verse.   In particular, $\| P - P_0\|<1$ implies that $\dim{\calP}=\dim{\calP_0}$.
\end{prop}
\begin{proof}
For any choice of the projectors $P,P_0$ the spectrum of $P-P_0$ lies on the interval $[-1,1]$
which implies $\|P-P_0\|\le 1$. Suppose $\|P-P_0\|=1$. Then there exists $\psi$
such that $(P-P_0)\, |\psi\ra =\pm |\psi\ra$. This is possible only if
$\la \psi | P | \psi\ra =1$, $\la \psi|P_0|\psi\ra =0$ or vice verse. Thus $\psi\in \calP$ and $\psi\in \calP_0^\perp$,
or vice verse. Conversely, if such a vector $\psi$ exists, then $(P-P_0)\, |\psi\ra =\pm |\psi\ra$, that is,
$\|P-P_0\|=1$. By the same token, if we assume that $\dim{\calP}>\dim{\calP_0}$ (or vice verse) there must exist
a vector $\psi\in \calP$ that is orthogonal to $\calP_0$ (or vice verse), that is, $\|P-P_0\|=1$.
\end{proof}

For any linear subspace $\calP$
define a reflection operator $R_{\calP}$ that flips  the sign of all vectors in $\calP$ and acts trivially on the orthogonal
complement to $\calP$, that is,
\be
\label{Refl1}
R_{\calP}=2P-I.
\ee
Following~\cite{DavisKahan69} let us define a direct rotation between a pair of subspaces
as follows.
\begin{dfn}
Let $\sqrt{z}$ be the square-root function defined on a complex plane with a branch cut along the negative real axis
and such that $\sqrt{1}=1$.
A unitary operator
\be
\label{U}
U=\sqrt{R_{\calP_0} R_\calP}
\ee
is called a direct rotation from $\calP$ to $\calP_0$.
\end{dfn}
\begin{lemma}
\label{lemma:d>1}
Suppose $\|P-P_0\|<1$.
Then no eigenvalue of $R_{\calP_0} R_{\calP}$ lies  on the negative real axis, so that $U$ is uniquely defined by Eq.~(\ref{U})
and
\be
\label{rotation}
U P  U^\dag =P_0.
\ee
\end{lemma}
It is worth pointing out that the direct rotation $U$ can also be defined as
the `minimal'  rotation that maps $\calP$ to $\calP_0$. More specifically,
among all unitary operators $V$ satisfying $VPV^\dag = P_0$ the direct rotation  differs least  from the identity in the Frobenius norm, see~\cite{DavisKahan69}.
If, in addition, $\| P-P_0\| <\sqrt{3}/2$, the Frobenius norm can be replaced by the operator norm~\cite{DavisKahan69}.
\begin{proof}[\bf Proof of Lemma~\ref{lemma:d>1}]
It is well-known that any pair of projectors can be simultaneously block-diagonalized with
blocks of size $2\times 2$ and $1\times 1$, and such that all $2\times 2$ blocks are rank-one
projectors, see Chapter~7.1 of~\cite{Bhatia}.
It follows that all operators $P,P_0,R_{\calP},R_{\calP_0}, U$ can be simultaneously block-diagonalized
with blocks of size $2\times 2$ and $1\times 1$.

Let us first consider some $1\times 1$ block. The restriction of $P$ and $P_0$ on this block
are scalars.  The condition $\|P-P_0\|<1$  implies that  $P=P_0=0$ or $P=P_0=1$. Thus
$R_{\calP}=R_{\calP_0}=\pm 1$.
In both case $R_{\calP_0} R_{\calP} =1$ and thus $U=1$.
Thus the identity $UPU^\dag=P_0$  holds for any $1\times 1$ block.

Let us now consider some $2\times 2$ block. In this block one has
$P=|\psi\ra\la \psi|$ and $P_0=|\phi\ra\la \phi|$ for some one-qubit states
$\psi,\phi \in \CC^2$. The condition $\|P-P_0\|<1$ implies that
$\psi$ and $\phi$ are non-orthogonal. In addition, the restriction of $U$
onto the considered block coincides with the direct rotation
$U_{\psi\to\phi}$, see Eq.~(\ref{canon1d}). Therefore Lemma~\ref{lemma:1d} implies that
$UPU^\dag =P_0$ in any $2\times 2$ block.
\end{proof}
For the later use let us point out one extra property of the direct rotation
which follows directly from the above proof and Corollary~\ref{cor:1d}.
\begin{corol}
\label{cor:d>1}
The direct
rotation from $\calP$ to $\calP_0$ can be written as
$U=\exp{(S)}$ where $S$ is an anti-hermitian operator
with the following properties:
\begin{itemize}
\item $PSP=P_0 S P_0=(I-P)S(I-P)=(I-P_0)S(I-P_0)=0$,
\item $\|S\|<\pi/2$.
\end{itemize}
\end{corol}
\begin{proof}
Indeed, apply Corollary~\ref{cor:1d} to each block in the decomposition of $P,P_0$, and $U$.
\end{proof}

\subsection{Generator of the direct rotation}
\label{subs:S}
The explicit formula Eq.~(\ref{U}) for the direct rotation $U$ is not very useful if one needs to compute
$U$ perturbatively. In this section we describe an alternative definition of the direct rotation
in terms of its generator, that is, an antihermitian operator $S$ such that $U=\exp{(S)}$.
We shall see later that the Taylor coefficients in the perturbative series for $S$ can be computed
using a simple inductive formula.
\begin{lemma}
\label{lemma:S}
Suppose $\|P-P_0\|<1$. Then there exists a unique
anti-hermitian operator $S$ with the following properties:\\
(1) $\exp{(S)} P \exp{(-S)} = P_0$,\\
(2) $S$ is block-off-diagonal with respect to $P_0$, that is,
\be
P_0 S P_0 =0 \quad \mbox{and} \quad (I-P_0) S (I-P_0)=0.
\ee
(3) $\|S\|<\pi/2$.\\
The unitary operator $U=\exp{(S)}$ coincides with the direct rotation from $\calP$ to
$\calP_0$.
\end{lemma}
\begin{proof}
We already know that  there exists at least one operator $S$
with the desired properties, see
Corollary~\ref{cor:d>1}.
Let us  show that   it is unique.
Indeed, write $S$ as
\be
\label{SW1}
S=\left( \ba{cc} 0 & S_{1,2}  \\
                             -S_{1,2}^\dag & 0 \\
                               \ea \right),
\ee
where the first and the second block correspond to the subspaces $\calP_0$ and $\calP_0^\perp$
respectively.
Computing the exponent $\exp{(S)}$ yields
\[
U=\exp{(S)}=\left( \ba{cc} U_{1,1} & U_{1,2}  \\
                             -U_{1,2}^\dag & U_{2,2} \\
                               \ea \right)
\]
where
\[
U_{1,1}=\cos{(A)}, \quad A\equiv \sqrt{S_{1,2} S_{1,2}^\dag} \quad
\mbox{and} \quad
U_{2,2}=\cos{(B)}, \quad B=\sqrt{S_{1,2}^\dag S_{1,2}}.
\]
Note that $\|A\| = \| B\| = \|S\| <\pi/2$. Hence $U_{1,1}$ and $U_{2,2}$ are hermitian positive-definite operators,
\be
\label{pos}
U_{1,1}>0 \quad \mbox{and} \quad U_{2,2}>0.
\ee
Let $\tilde{S}$ be any anti-hermitian operator satisfying the three conditions of the lemma
and let $\tilde{U}=\exp{(\tilde{S})}$. We have to prove that $\tilde{U}=U$ and $\tilde{S}=S$. Indeed, using the identity
\[
(\tilde{U} U^\dag)P_0(\tilde{U} U^\dag)^\dag=\tilde{U}P\tilde{U}^\dag =P_0
\]
we conclude that $\tilde{U} U^\dag$ commutes with $P_0$. This is possible only if $\tilde{U} U^\dag=L$
for some block-diagonal unitary $L$, that is,
\[
L=\left( \ba{cc} L_1 & 0  \\
                             0 & L_2 \\
                               \ea \right), \quad L_1 L_1^\dag=I, \quad L_2^\dag L_2=I.
\]
It follows that $\tilde{U}=LU$, that is $\tilde{U}_{1,1}=L_1 U_{1,1}$ and $\tilde{U}_{2,2}=L_2 U_{2,2}$.
Combining it with Eq.~(\ref{pos}) we arrive at $L_1=I$ and $L_2=I$ as follows from the following proposition.
\begin{prop}
\label{prop:1}
Let $H$ be a positive operator and $L$ be a unitary operator.
Suppose that $L H$ is also a positive operator. Then $L=I$.
\end{prop}
\begin{proof}
Indeed, consider the eigenvalue decomposition $L=\sum_{\alpha} e^{i\theta_\alpha} |\alpha\ra\la \alpha|$.
Since $L H>0$ one gets $\la \alpha |L H |\alpha\ra =  e^{i\theta_\alpha} \la \alpha |H|\alpha\ra > 0$.
Since $\la \alpha |H|\alpha\ra >0$ we conclude that $e^{i\theta_\alpha}>0$, that is, $e^{i\theta_\alpha}=1$.
Hence $L=I$.
\end{proof}
To summarize, we have shown that $\tilde{U}=U$.
Since the matrix exponential function $\exp{(M)}$ is invertible on the subset of matrices
satisfying $\|M\|<\pi/2$, we conclude that $\tilde{S}=S$.
\end{proof}

\subsection{Weak multiplicativity of the direct rotation}
\label{subs:mult}
Consider a bipartite system of Alice and Bob with a Hilbert space
$\calH=\calH^A \otimes \calH^B$.
Let $P^A,P^A_0$ be a pair of Alice's projectors acting on $\calH^A$.
Similarly, let $P^B,P^B_0$ be a pair of Bob's projectors acting on $\calH^B$.
We shall assume that
\be
\label{AcloseBclose}
\| P^A - P^A_0\| <1 \quad \mbox{and} \quad \|P^B-P^B_0\|<1
\ee
such that one can define local direct rotations $U^{A}$ and $U^{B}$,
\[
U^{A} P^A (U^{A})^\dag =P_0^A \quad \mbox{and} \quad U^{B} P^B (U^{B})^\dag = P_0^B.
\]
Consider also the global direct rotation $U^{AB}$ mapping $P^A\otimes P^B$ to $P^A_0 \otimes P^B_0$, that is,
\[
U^{AB} (P^A \otimes P^B) (U^{AB})^\dag = P^A_0\otimes P^B_0.
\]
It is worth pointing out that  in general $U^{AB}\ne U^{A} \otimes U^{B}$.
Indeed,  consider a special case when $P^A$ and $P^B$ project onto
one-dimensional subspaces of Alice and Bob.
Then clearly $P^A\otimes P^B$ is a rank-one projector. In this case $U^{AB}$ is a rotation in  some two-dimensional subspace of $\calH$, see Section~\ref{subs:1d}. On the other hand $U^{A} \otimes U^{B}$ is
a tensor product of a rotation in some two-dimensional subspace of $\calH^A$
and a rotation in some two-dimensional subspace of $\calH^B$. Therefore $U^{A}\otimes U^{B}$
acts non-trivially on some four-dimensional subspace of $\calH$ and hence
$U^{AB}\ne U^{A}\otimes U^{B}$.
This example shows that the direct rotation is not multiplicative under the tensor product.
Nevertheless, it features a certain weaker form of multiplicativity which  we derive in this section.

To start with, we need to check that the direct rotation $U^{AB}$ is well-defined.
It follows from the following proposition.
\begin{prop}
\label{prop:ABclose}
Suppose Alice's projectors  $P^A,P^A_0$ and Bob's projectors $P^B,P^B_0$  satisfy non-orthogonality constraints
Eq.~(\ref{AcloseBclose}).
Then
\be
\label{ABclose}
\| P^A \otimes P^B - P^A_0 \otimes P^B_0 \|<1.
\ee
\end{prop}
\begin{proof}
Let us assume that Eq.~(\ref{ABclose}) is false. Then there exists a state $\psi\in \calH$
such that $P^A \otimes P^B \, |\psi\ra =0$ and $P^A_0\otimes  P^B_0\, |\psi\ra=|\psi\ra$
(or vice verse). Below we show that it leads to a contradiction.
Indeed, the condition $\|P^A - P^A_0\|<1$ implies that there exists $\alpha>0$ such that
$P^A - P^A_0 \le (1-\alpha)\, I$. Multiplying this inequality on both sides by $P^A$ we arrive at
$P^A - P^A P^A_0 P^A \le (1-\alpha)P^A$, that is, $P^A P^A_0 P^A\ge \alpha P^A$.
Similarly, there exists $\beta>0$ such that $P^B P^B_0 P^B\ge \beta P^B$.
Thus
\be
\label{ABABAB}
(P^A\otimes P^B) (P^A_0\otimes P^B_0 ) (P^A\otimes P^B) = (P^A P^A_0 P^A)\otimes (P^B P^B_0 P^B)
\ge \alpha\beta P^A_0\otimes P^B_0.
\ee
Here we used the fact that a tensor product of positive semi-definite operators
is a positive semi-definite operator.
Computing the expectation value of Eq.~(\ref{ABABAB}) on $\psi$ one gets
$0\ge \alpha \beta$ which is a contradiction.
\end{proof}
We conclude that the global direct rotation $U^{AB}$ mapping $P^A\otimes P^B$
to $P^A_0 \otimes P^B_0$ is well-defined.
The relationship between the global and the local rotations which we shall call
a {\em weak multiplicativity} is established by the
following lemma.
\begin{lemma}[\bf Weak multiplicativity]
\label{lemma:mult}
Let $U^{A},U^{B}$, and $U^{AB}$ be the direct rotations defined above. Then
\be
\label{mult}
U^{AB}\, (P^A \otimes P^B) = (U^{A} \otimes U^{B} ) (P^A\otimes  P^B).
\ee
\end{lemma}
\begin{proof}
Let us perform the simultaneous block-diagonalization of $P^A,P^A_0$ and $P^B,P^B_0$ with blocks
of size $2\times 2$ and $1\times 1$, as in the proof of Lemma~\ref{lemma:d>1}.
Recall that each $2\times 2$ block is a rank-one projector. Hence it suffices
to check Eq.~(\ref{mult}) only for the case when $\calH^A =\calH^B=\CC^2$
and rank-one projectors  $P^A,P^B,P^A_0,P^B_0$.
Choose  normalized states $\psi^A$, $\psi^B$, $\psi^A_0$,
$\psi^B_0$ in the range of the above projectors, and fix their relative phase such that $\la \psi^A|\psi^A_0\ra>0$
and $\la \psi^B|\psi^B_0\ra>0$. Lemma~\ref{lemma:1d} implies that
$U^{A}\, |\psi^A\ra=|\psi^A_0\ra$ and $U^{B} \, |\psi^B\ra =|\psi^B_0\ra$.
However, since $\la \psi^A \otimes \psi^B|\psi^A_0\otimes \psi^B_0\ra>0$,
Lemma~\ref{lemma:1d} also implies that $U^{AB}\, |\psi^A \otimes \psi^B\ra = |\psi^A_0\otimes \psi^B_0\ra$.
This we have an identity $U^{AB}(P^A \otimes P^B)=(U^{A} \otimes U^{B}) (P^A\otimes  P^B)$ in each tensor product of $2\times 2$ blocks. Similar arguments hold if one or both blocks have size $1\times 1$.
\end{proof}

\section{Effective low-energy Hamiltonian}
\label{sec:SW}

\subsection{Schrieffer-Wolff transformation}
Let $\calH$ be a finite-dimensional Hilbert space and
$H_0$ be a hermitian operator on $\calH$.
We shall refer to $H_0$  as an {\em unperturbed Hamiltonian}.
Let $\calI_0\subseteq \RR$ be any  interval containing one or several eigenvalues of $H_0$
and let $\calP_0\subseteq \calH$ be the subspace spanned by all eigenvectors of $H_0$ with eigenvalue
lying in $\calI_0$.
We shall say that $H_0$ has a {\em spectral gap} $\Delta$ iff
for any pair of eigenvalues $\lambda,\eta$ such that $\lambda\in \calI_0$ and $\eta\notin \calI_0$ one has $|\lambda-\eta|\ge \Delta$. In other words, the eigenvalues of $H_0$ lying in $\calI_0$
must be separated from the rest of the spectrum by a gap at least $\Delta$.

Consider now a perturbed Hamiltonian $H=H_0+\epsilon \, V$, where the perturbation $V$ is an arbitrary hermitian
operator on $\calH$.
Let $\calI \subseteq \RR$ be the interval obtained from  $\calI_0$ by adding margins of thickness $\Delta/2$
on the left and on the right of $\calI_0$.
 Let $\calP\subseteq \calH$ be the subspace spanned by all eigenvectors of $H$ with eigenvalue
lying in $\calI$.
In this section we shall only consider sufficiently weak perturbations that do no close the gap
separating the interval $\calI_0$ from the rest of the spectrum. Specifically, we shall assume that
$|\epsilon|\le \epsilon_c$, where
\be
\label{epsilon_c}
\epsilon_c=\frac{\Delta}{2\|V\|}.
\ee
Since the perturbation shifts any eigenvalue at most by $\|\epsilon V\|$, we conclude that
the eigenvalues of $H$ lying in $\calI$ are separated from the rest of the spectrum of $H$
by a positive gap as long as $|\epsilon|<\epsilon_c$.
In particular, it implies that $\calP_0$ and $\calP$ have the same dimension.
Let $P_0$ and $P$ be the orthogonal projectors onto $\calP_0$ and $\calP$.
Introduce also projectors $Q_0=I-P_0$ and $Q=I-P$.
\begin{lemma}
\label{lemma:per1}
Suppose $\epsilon$ is real and $|\epsilon|<\epsilon_c$. Then $\|P_0-P\|\le 2\|\epsilon V\|/\Delta <1$.
\end{lemma}
\begin{proof}
Let $T=P-P_0$. A simple algebra shows that $T^2$ is block-diagonal with respect
to $P_0$ and $Q_0$, that is $P_0 T^2 Q_0=0$ and $Q_0 T^2 P_0=0$.
The remaining diagonal blocks of $T^2$ are
$P_0 T^2 P_0=P_0 Q P_0$ and $Q_0 T^2 Q_0 = Q_0 P Q_0$.
Since $T$ is hermitian, we have $\|T\| = \sqrt{\|T^2\|}$ and thus
\[
\|T\| = \max{(\sqrt{\| P_0 QP_0\|},\sqrt{\|Q_0 P Q_0\|})}= \max{(\| P_0 Q\| ,\|Q_0 P\|)}.
\]
Let us assume that $\|T\|=\|P_0Q\|$ (the case $\|T\|=\|PQ_0\|$ is dealt with analogously\footnote{It is worth mentioning
that $\|P-P_0\|=\|P_0 Q\|=\|Q_0 P\|$ for any projectors $P,P_0$ satisfying $\|P-P_0\|<1$. These identities can be
easily proved using the simultaneous block-diagonalized form of $P$ and $P_0$.}).
Define auxiliary operators $A=P_0 H_0 P_0$,  $B=QHQ$,
and $X=P_0 Q$. Then one has
\be
\label{Sylv}
AX-XB=P_0 (H_0-H)Q = -\epsilon P_0 V Q\equiv Y.
\ee
The equation $AX-XB=Y$ is known as the Sylvester equation. In particular, it is known
that $\|X\| \le \delta^{-1} \|Y\|$ whenever spectrums of $A$ and $B$ can be separated by
an annulus of width $\delta$ in the complex plane, see for instance, Theorem~VII.2.11 in~\cite{Bhatia}.
In our case the spectrums of $A$ and $B$ considered as operators on $P_0$ and $Q$ respectively
are separated by an annulus of width $\Delta/2$ whenever $|\epsilon|<\epsilon_c$. Thus we arrive at
\[
\|X\| \le (2/\Delta) \, \| Y\| \le  (2/\Delta)\, \| \epsilon V\| <1
\]
as long as $|\epsilon|<\epsilon_c$.
\end{proof}
Lemma~\ref{lemma:per1} implies that the direct rotation $U$ from $\calP$ to $\calP_0$
is well-defined as long as $|\epsilon|<\epsilon_c$. By definition, it satisfies $UPU^\dag=P_0$
and $UQU^\dag =Q_0$.
In addition, since the perturbed Hamiltonian
$H=H_0+\epsilon V$ is block-diagonal with respect to $P$ and $Q$, we conclude that the transformed
Hamiltonian $UHU^\dag$ is block-diagonal with respect to $P_0$ and $Q_0$.
In almost all applications of perturbation theory the interval $\calI_0$ includes only the smallest eigenvalue
of $H_0$ or all sufficiently small eigenvalues. By a slight abuse of notations we shall
adopt the standard terminology and refer to $\calP_0$ as the {\em low-energy}
subspace of $H_0$. However, one should keep in mind that all results of this section remain valid
for an arbitrary choice of the interval $\calI_0$.
\begin{dfn}
The direct rotation $U$  from $\calP$ to $\calP_0$ is called the Schrieffer-Wolff (SW) transformation
for the unperturbed Hamiltonian $H_0$, perturbation $\epsilon V$, and the low-energy subspace $\calP_0$.
The operator $H_{\eff}=P_0 U (H_0+\epsilon\, V)U^\dag P_0$
is called an effective low-energy Hamiltonian.
\end{dfn}

\subsection{Derivation of the perturbative series}
\label{subs:series}
Let $U$ be the Schrieffer-Wolff transformation constructed for some unperturbed Hamiltonan $H_0$,
a perturbation $\epsilon V$, and the low-energy subspace $\calP_0$. In this section we always assume that
$|\epsilon|<\epsilon_c$, so $U$ is well-defined, see Lemma~\ref{lemma:per1}.
In most of applications  the explicit formula Eq.~(\ref{U}) for the SW transformation cannot be used
directly because  the low-energy subspace $\calP$ of the perturbed Hamiltonian  $H_0+\epsilon V$
is unknown. In this section we explain how one can  compute the transformation $U$ and the effective low-energy
Hamiltonian $H_{\eff}$ perturbatively. Truncating the series for $H_{\eff}$ at some finite order one obtains an
effective low-energy theory describing properties of the perturbed system.
Throughout this section all series are treated as formal series. Their convergence will be proved
later in Section~\ref{subs:convergence}.

We begin by introducing some notations. The space of linear operators acting on $\calH$ will be denoted
$\LL(\calH)$. A linear map $\calO\,: \, \LL(\calH)\to \LL(\calH)$ will be referred to as a {\em superoperator}.
We shall often use  superoperators
\[
\calO(X)=P_0 X Q_0 + Q_0 X P_0 \quad  \mbox{and} \quad \calD(X)=P_0 X P_0 + Q_0 X Q_0.
\]
We shall say that an operator $X$ is {\em block-off-diagonal}  iff $\calO(X)=X$.
An operator $X$ is called {\em block-diagonal} iff $\calD(X)=X$.
Decompose the perturbation $V$ as $V=\Vd+\Vod$, where
\[
\Vd=\calD(V) \quad \mbox{and} \quad \Vod=\calO(V)
\]
are block-diagonal and block-off-diagonal parts of $V$.
Given any operator $Y\in \LL(\calH)$ define a superoperator  $\hat{Y}$
describing the adjoint action of $Y$, that is, $\hat{Y}(X)=[Y,X]$.

Recall that the SW transformation can be uniquely represented as $U=\exp{(S)}$ where
$S$ is an anti-hermitian generator which is block-off-diagonal and $\|S\|<\pi/2$, see Lemma~\ref{lemma:S},
whereas the transformed Hamiltonian $e^S(H_0+\epsilon V)e^{-S}$ is block-diagonal.
Combining these conditions would yield Taylor series for $S$, which, in turn,
yields Taylor series for $H_{\eff}$. We begin by rewriting the transformed Hamiltonian as
\be
\label{transformed}
\exp{(\hat{S})}(H_0+\epsilon V)=\cosh{(\hat{S})} (H_0+\epsilon \Vd) + \sinh{(\hat{S})}(\epsilon \Vod)+
\sinh{(\hat{S})}(H_0 +\epsilon \Vd) + \cosh{(\hat{S})}(\epsilon \Vod).
\ee
Taking into account that $S$ is block-off-diagonal,  we conclude that the first and the second
terms in the righthand side of Eq.~(\ref{transformed}) are block-diagonal, while the third and the fourth terms are block-off-diagonal.
In order for the transformed Hamiltonian to be block-diagonal, $S$ must obey
\[
\sinh{(\hat{S})}(H_0 +\epsilon \Vd) + \cosh{(\hat{S})}(\epsilon \Vod)=0.
\]
Since our goal is to derive formal Taylor series, $S$ can be regarded as an infinitesimally small operator.
In this case  the superoperator $\cosh{(\hat{S})}$ is invertible and thus the above condition
can be rewritten as
\be
\label{tanhS}
\tanh{(\hat{S})}(H_0+\epsilon \Vd) + \epsilon \Vod =0.
\ee
Now we can rewrite the transformed Hamiltonian as
\bea
\exp{(\hat{S})}(H_0+\epsilon V) &=& \cosh{(\hat{S})} (H_0+\epsilon \Vd) + \sinh{(\hat{S})}(\epsilon \Vod) \nn \\
&=& H_0+\epsilon \Vd + (\cosh{(\hat{S})} - 1)(H_0+\epsilon \Vd) +  \sinh{(\hat{S})}(\epsilon \Vod) \nn \\
&=& H_0+\epsilon \Vd +\frac{(\cosh{(\hat{S})} - 1)}{\tanh{(\hat{S})}} \tanh{(\hat{S})} (H_0+\epsilon \Vd)  +  \sinh{(\hat{S})}(\epsilon \Vod) \nn.
\eea
Note that $(\cosh{(\hat{S})} - 1)/\tanh{(\hat{S})}$ is well defined for infinitesimally small $S$
by its Taylor series. Using Eq.~(\ref{tanhS}) we arrive at
\[
\exp{(\hat{S})}(H_0+\epsilon V) = H_0+\epsilon \Vd  + F(\hat{S})(\epsilon \Vod),
\]
where
\[
F(x) = \sinh{(x)} - \frac{\cosh{(x)} -1}{\tanh{(x)}}.
\]
A simple algebra shows that $F(x)=\tanh{(x/2)}$, so we finally get
\be
\label{transformed1}
\exp{(\hat{S})}(H_0+\epsilon V)  = H_0 + \epsilon \Vd +  \tanh{(\hat{S}/2)}(\epsilon \Vod).
\ee
In order to solve Eq.~(\ref{tanhS}) for $S$ let us introduce some more notations.
Let $\{|i\ra\}$ be an orthonormal eigenbasis of $H_0$ such that $H_0\,|i\ra = E_i\, |i\ra$ for
all $i$. We shall use notation $i\in \calI_0$ as a shorthand for $E_i\in \calI_0$.
Define a superoperator
\be
\label{L0inv}
\calL(X)=\sum_{i,j}
\frac{\la i|\calO(X)|j\ra}{E_i - E_j}\, |i\ra\la j|.
\ee
Note that $\la i|\calO(X)|j\ra=0$ whenever  $i,j\in \calI_0$ or  $i,j\notin \calI_0$.
Let us agree that the sum in Eq.~(\ref{L0inv}) includes only the terms
with $i\in \calI_0$, $j\notin \calI_0$ or vice verse. In this case the energy denominator
can be bounded as $|E_i-E_j|\ge \Delta$.
One can easily check that
\be
\label{L0inv'}
\calL([H_0,X]) = [H_0,\calL(X)] =\calO(X)
\ee
for any operator $X\in \LL(\calH)$. It is worth mentioning that $\calL$
maps hermitian operators to anti-hermitian operators and vice verse.

Treating $S$ as an  infinitesimally small operator we can rewrite
Eq.~(\ref{tanhS}) as
\[
\hat{S}(H_0+\epsilon \Vd) + \hat{S} \coth{(\hat{S})}(\epsilon \Vod)=0.
\]
Using Eq.~(\ref{L0inv'}) and the fact that $S$ is block-off-diagonal one gets
\be
\label{Seq1}
S=\calL\,  \hat{S}(\epsilon \Vd) + \calL\,  \hat{S} \coth{(\hat{S})}(\epsilon \Vod).
\ee
We can solve Eq.~(\ref{Seq1}) in terms of Taylor series,
\be
\label{Sseries}
S=\sum_{n=1}^\infty S_n\, \epsilon^n, \quad S_n^\dag=-S_n.
\ee
Let us also  agree that $S_0=0$. We shall need  the Taylor series
\be
\label{coth}
x \coth{(x)} = \sum_{n=0}^\infty a_{2n} x^{2n}, \quad  a_{m} = \frac{2^{m} B_{m}}{m!},
\ee
where $B_{m}$ are the Bernoulli numbers.
Then we have
\bea
S_1 &=& \calL(\Vod), \nn \\
S_2 &=& -\calL \hat{\Vd} (S_1), \nn \\
S_n &=& -\calL\hat{\Vd}(S_{n-1})  + \sum_{j\ge 1}  a_{2j} \calL\,  \hat{S}^{2j} (\Vod)_{n-1} \quad \quad \mbox{for $n\ge 3$}.
\label{Seq2}
\eea
Here we used a shorthand
\be
\label{Seq3}
\hat{S}^k(\Vod)_{m} = \sum_{\substack{n_1,\ldots,\, n_k\ge 1\\ \\ n_1 + \ldots + n_k =m\\}} \; \; \hat{S}_{n_1} \cdots \hat{S}_{n_k}(\Vod).
\ee
Note that the righthand side of Eq.~(\ref{Seq2}) depends only on $S_1,\ldots,S_{n-1}$.
Hence Eq.~(\ref{Seq2}) provides an inductive rule for computing the Taylor coefficients $S_n$.
Projecting Eq.~(\ref{transformed1}) onto the low-energy subspace one gets
\be
\label{Heff}
H_{\eff}=H_0 P_0 + \epsilon P_0 V P_0 + \sum_{n=2}^\infty \epsilon^n H_{\eff,n},
 \quad H_{\eff,n} = \sum_{j\ge 1}  b_{2j-1} P_0 \, \hat{S}^{2j-1}(\Vod)_{n-1} P_0.
\ee
Here $b_{2j-1}$  are the Taylor coefficients of the function $\tanh{(x/2)}$.
More explicitly,
\be
\label{tanh}
\tanh{(x/2)} = \sum_{n=1}^\infty b_{2n-1} x^{2n-1}, \quad b_{2n-1} = \frac{2 (2^{2n}-1) B_{2n}}{(2n)!}.
\ee
Note that $H_{\eff,n}$ depends only upon $S_1,\ldots,S_{n-1}$.
To illustrate the method let us compute the Taylor coefficients $H_{\eff,n}$ for small values of $n$
(a systematic way of computing $H_{\eff,n}$ in terms of diagrams is described in the next section).
From Eq.~(\ref{Seq2}) one easily finds
\bea
S_3 &=& -\calL \hat{\Vd} (S_2)  + a_2 \calL\, \hat{S}_1^2(\Vod), \nn \\
S_4 &=& -\calL\hat{\Vd}(S_3)  + a_2 \calL\, (\hat{S}_1 \hat{S}_2 + \hat{S}_2 \hat{S}_1)(\Vod). \nn
\eea
From Eq.~(\ref{Heff}) one gets
\bea
H_{\eff,2} &=& b_1 P_0 \hat{S}_1(\Vod) P_0, \label{H1} \\
H_{\eff,3} &=& b_1 P_0 \hat{S}_2(\Vod) P_0, \label{H2} \\
H_{\eff,4} &=& b_1 P_0 \hat{S}_3(\Vod) P_0 + b_3 P_0 \hat{S}_1^3(\Vod) P_0, \label{H3} \\
H_{\eff,5} &=& b_1 P_0 \hat{S}_4(\Vod) P_0 + b_3 P_0 (\hat{S}_2 \hat{S}_1^2 +
\hat{S}_1 \hat{S}_2 \hat{S}_1 + \hat{S}_1^2 \hat{S}_2)(\Vod) P_0.\label{H4}
\eea
Let us obtain explicit formulas for $H_{\eff,3}$ and $H_{\eff,4}$ that include only $S_1$.
Using the identity $\hat{S}_2(\Vod)=-\hat{V}_{\mathrm{od}}(S_2)$ we get
\be
\label{Heff3}
H_{\eff,3}=b_1 P_0 \hat{V}_{\mathrm{od}} \calL \hat{\Vd}(S_1) P_0
\ee
and
\[
S_3=(\calL \hat{V}_{\mathrm{d}})^2(S_1)  + a_2 \calL \hat{S}_1^2(\Vod).
\]
Substituting $S_3$ into the expression for $H_{\eff,4}$ and  using the identity
$\hat{S}_3(\Vod)=-\hat{V}_{\mathrm{od}}(S_3)$ one gets
\be
\label{Heff4}
H_{\eff,4}=P_0\left[ -b_1 \hat{V}_{\mathrm{od}} (\calL \hat{V}_d)^2(S_1)   -
b_1 a_2 \hat{V}_{\mathrm{od}} \calL \hat{S}_1^2(\Vod)  +
b_3 \hat{S}_1^3 (\Vod)\right] P_0.
\ee
The explicit values of the Taylor coefficients $a_{2n}$ and $b_{2n-1}$
are listed in Table~\ref{fig:ab}.
\begin{figure}
\centerline{
\begin{tabular}{|c|c|c|c|}
\hline
$a_0$ & $a_2$ & $a_4$ & $a_6$ \\
\hline
$1$ & $1/3$ & $-1/{45}$ &  ${2}/{945}$ \\
\hline
\end{tabular}
\ \ \ \
\begin{tabular}{|c|c|c|c|}
\hline
$b_1$ & $b_3$ & $b_5$ & $b_7$ \\
\hline
$1/2$ & $-1/24$ & $1/{240}$ &  $-{17}/{40320}$ \\
\hline
\end{tabular}
}
\caption{The value of the Taylor coefficients $a_{2n}$ and $b_{2n-1}$ for small $n$.}
\label{fig:ab}
\end{figure}
The formula for $H_{\eff,4}$ can be slightly simplified in  the special case
when $\calI_0$ contains a single eigenvalue of $H_0$, that is, the restriction of $H_0$ onto
$P_0$ is proportional to the identity operator. In this special case one can use an identity
\be
\label{nice_ident}
P_0 [\calL(X),\calO(Y)]P_0 =-P_0[\calO(X),\calL(Y)]P_0
\ee
which holds for any operators $X,Y$. One can easily check Eq.~(\ref{nice_ident}) by computing
matrix elements $\la i|\cdot |j\ra$ of both sides for $i,j\in \calI_0$.
Using Eq.~(\ref{nice_ident}) and explicit values of $a$'s and $b$'s one can rewrite $H_{\eff,4}$ as
\be
\label{Heff4simple}
H_{\eff,4}=P_0 \left[
\frac18 \hat{S}_1^3(\Vod) -\frac12 \hat{V}_{\mathrm{od}} (\calL \hat{V}_d)^2(S_1)   \right] P_0.
\ee
To summarize, the Taylor series for the effective low-energy Hamiltonian can be written as
\bea
H_{\eff} &=& H_0 P_0 + \epsilon P_0 V P_0 + \frac{\epsilon^2}2   P_0 \hat{S}_1(\Vod) P_0 +
\frac{\epsilon^3 }2 P_0 \hat{V}_{\mathrm{od}} \calL \hat{\Vd} (S_1) P_0 \nn \\
&& - {\epsilon^4} P_0 \left[
\frac12 \hat{V}_{\mathrm{od}} (\calL \hat{\Vd} )^2(S_1) +
\frac16 \hat{V}_{\mathrm{od}} \calL \hat{S}_1^2 (\Vod)+
\frac1{24} \hat{S}_1^3(\Vod)
\right] P_0 + O(\epsilon^5),
\eea
with the simplification Eq.~(\ref{Heff4simple}) in the case when the restriction of $H_0$
onto the low-energy subspace is proportional to the identity operator.

\subsection{Schrieffer-Wolff diagram technique}
\label{subs:diagram}
In this section we develop a diagram technique for the effective low-energy Hamiltonian.
By analogy with the Feynman-Dyson diagram technique, we shall represent
the $n$-th order Taylor coefficient $H_{\eff,n}$ as a weighted sum over certain class of admissible diagrams. In our case $n$-th order diagrams will be trees with $n$ nodes obeying certain restrictions imposed on the node degrees.
Any such diagram represents a linear
operator acting on the low-energy subspace $\calP_0$. Loosely speaking, every edge of the tree stands for a commutator, while every node of the tree stands for $\Vd$ or $\Vod$. The weight associated with a diagram depends only
on node degrees and can be easily expressed in terms of the Bernoulli numbers.

Let us now proceed with formal definitions. Let $T$ be a connected tree
with $n$ nodes. Each node $u\in T$ has at most one parent node and
zero or more children nodes. The children of any node are ordered in some fixed way.
We shall denote children of a node $u$ as $u(1),\ldots,u(k)$, where $k$ is the number of children  (NOC) of $u$.
Nodes having no children are called {\em leaves}.
There is a unique {\em root} node  which has no parent node.

For any node $u\in T$ we shall define a linear operator $O_u$ acting on the full Hilbert space
$\calH$. The definition is inductive, so we shall start from the leaves and move up towards the root.
\begin{itemize}
\item If $u$ is a leaf then $O_u=S_1\equiv \calL(\Vod)$.
\item If $u\ne \mathrm{root}$ has $k\ge 2$ children and $k$ is even then  $O_u=\calL \, \hat{O}_{u(1)} \cdots \hat{O}_{u(k)}(\Vod)$.
\item If $u\ne \mathrm{root}$ has exactly one child  then $O_u=-\calL\hat{\Vd}(O_{u(1)})$.
\item If $u=\mathrm{root}$ has $k$ children with odd $k$ then $O_u=P_0 \, \hat{O}_{u(1)} \cdots \hat{O}_{u(k)}(\Vod) P_0$.
\item In all remaining cases $O_u=0$.
\end{itemize}
Recall that $\hat{O}\equiv \mathrm{ad}_O$ describes the adjoint action of $O$.
Define an operator $O(T)$ associated with $T$ as  $O(T)=O_r$, where $r$ is the root of $T$.
It is clear from the above definition that $O(T)$ is an operator acting on the low-energy
subspace $\calP_0$ and $O(T)$ has degree
$n$ in the perturbation $V$. Furthermore, $O(T)=0$ unless
the root of $T$ has odd NOC and every node different from the root
with more than one child has even NOC.
\begin{dfn}
A connected tree $T$ is called admissible iff
\begin{itemize}
\item The root of $T$ has odd NOC,
\item  Every node different from the root
has either even NOC or NOC=1.
\end{itemize}
\end{dfn}
A complete list of admissible trees (modulo isomorphisms) with $n=3,4,5,6$ nodes
is shown on Fig.~\ref{fig:diagrams}.
\begin{figure}
\centerline{
\includegraphics[height=10cm]{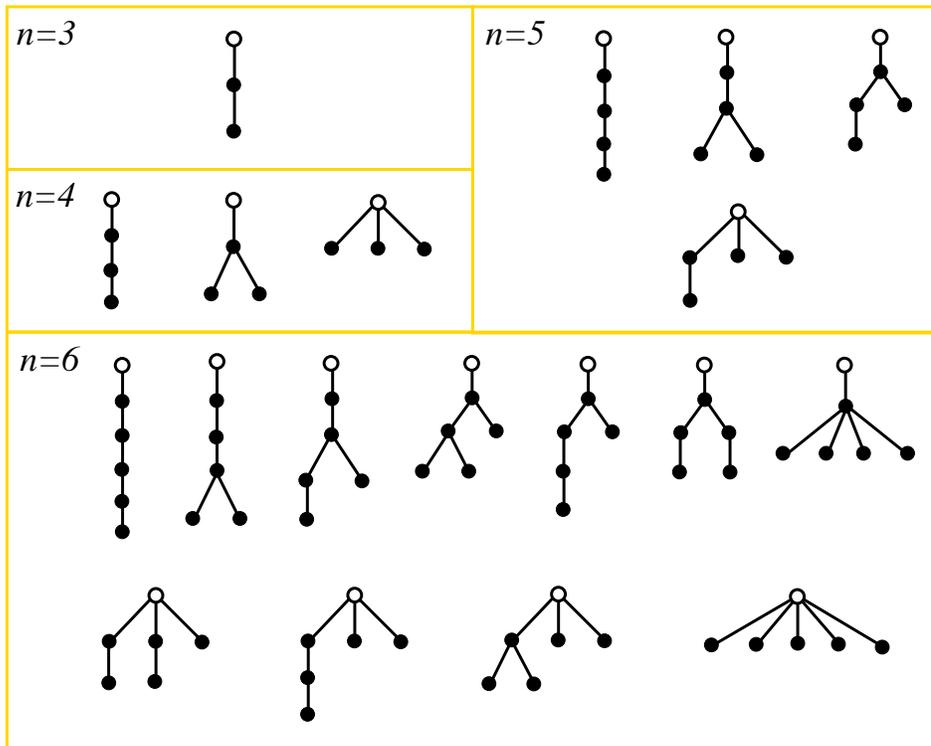}}
\caption{A set of admissible trees $\calT(n)$ with $n=3,4,5,6$ nodes. Only one representative
for each class of isomorphic trees is shown. The total number of admissible
trees (counting isomorphic trees) is $|\calT(n)|=1,3,7,20$ for $n=3,4,5,6$ respectively, }
\label{fig:diagrams}
\end{figure}
To illustrate the definition of operators $O(T)$ let us consider fourth-order diagrams.
Let   $T_1,T_2,T_3$ be the admissible trees with four nodes
shown on Fig.~\ref{fig:diagrams} (listed from the left to the right).
Then one has
\bea
O(T_1) &=&  -P_0 \hat{V}_{\mathrm{od}} (\calL \hat{\Vd} )^2(S_1) P_0,\nn \\
O(T_2) &=& -P_0 \hat{V}_{\mathrm{od}} \calL \hat{S}_1^2 (\Vod) P_0,\nn \\
O(T_3) &=& P_0 \hat{S}_1^3(\Vod)P_0. \label{T123}
\eea
For any tree $T$ define a weight
\be
\label{weight}
w(T)=\prod_{u\in T} w(u),
\ee
where the product is over all nodes of $T$ and
\be
w(u)=\left\{ \ba{rcl} 1 &\mbox{if} & \mbox{$u\ne \mathrm{root}$ has exactly one child}, \\
a_{k} &\mbox{if} & \mbox{$u\ne \mathrm{root}$ has $k$ children for some even $k$},\\
b_k &\mbox{if} & \mbox{$u=\mathrm{root}$ has $k$ children for some odd $k$},\\
0 && \mbox{otherwise}. \\
\ea
\right.
\ee
Recall that
\[
a_{2n} = \frac{2^{2n} B_{2n}}{(2n)!} \quad \mbox{and} \quad b_{2n-1} = \frac{2 (2^{2n}-1) B_{2n}}{(2n)!},
\]
where $B_{2n}$ are the Bernoulli numbers, see also Table~\ref{fig:ab}.
The main result of this section is the following lemma.
\begin{lemma}
The Taylor series for $H_{\eff}$ can be written as
\be
\label{Hdiag}
H_{\eff}= H_0 P_0 + \epsilon P_0 V P_0 + \sum_{n=2}^\infty \; \; \sum_{T\in \calT(n)} \epsilon^n \, w(T) \, O(T),
\ee
where $\calT(n)$ is the set of all admissible trees with $n$ nodes.
\end{lemma}
\begin{proof}
Let us first develop a diagram technique for the generator $S$.
Let $T$ be any connected tree.
For any node $u\in T$ we shall define a linear operator $O_u'$ acting on the full Hilbert space
$\calH$. The definition is inductive, so we shall start from the leaves and move up towards the root.
\begin{itemize}
\item If $u$ is a leaf then $O_u'=S_1\equiv \calL(\Vod)$.
\item If $u$ has exactly one child then $O_u'=-\calL\hat{\Vd}(O_{u(1)}')$.
\item If $u$ has $k$ children and $k$ is even then  $O_u'=\calL \, \hat{O}_{u(1)}' \cdots \hat{O}_{u(k)}'(\Vod)$.
\item In all remaining cases $O_u'=0$.
\end{itemize}
Define an operator $O'(T)$ associated with $T$ as  $O'(T)=O_r'$, where $r$ is the root of $T$.
It is clear from the above definition that $O'(T)$ is a block-off-diagonal operator and $O'(T)$ has degree
$n$ in the perturbation $V$. Furthermore, $O'(T)=0$ unless all nodes of $T$ with more than one child
have even NOC.
\begin{dfn}
A tree $T$ is called $S$-admissible iff any node of $T$ with more than one child has even number
of children.
\end{dfn}
Define a {\em weight} of a tree $T$ as
\be
\label{weight'}
w'(T)=\prod_{u\in T} w'(u),
\ee
where the product is over all nodes of $T$ and
\be
w'(u)=\left\{ \ba{rcl} 1 &\mbox{if} & \mbox{$u$ has exactly one child}, \\
a_{k} &\mbox{if} & \mbox{$u$ has $k$ children for some even $k$},\\
0 && \mbox{otherwise}. \\
\ea
\right.
\ee
\begin{lemma}
\label{lemma:Sdiag}
The generator of the SW transformation can be written as
\be
\label{Sdiag}
S=\sum_{n=1}^\infty \;  \; \sum_{T\in \calT'(n)} \epsilon^n\,  w'(T)\, O'(T).
\ee
where $\calT'(n)$ is the set of all $S$-admissible trees with $n$ nodes.
\end{lemma}
\begin{proof}
We can use induction in $n$
to show that $S_n$ is the sum of $w'(T) O'(T)$ over all $S$-admissible trees with $n$ nodes.
 The base of induction is $n=1$ in which case
there is only one $S$-admissible tree $T$ (a single node) and $O'(T)=S_1$.
To prove the induction hypothesis for general $n$ we can use Eq.~(\ref{Seq2}) and Eq.~(\ref{Seq3}).
Repeatedly expanding $\hat{S}^{2j}$  in Eq.~(\ref{Seq2}) one gets a sum over  $S$-admissible  trees, where
the operators $O_u'$ correspond to $S_{n_i}$ encountered in the course of the expansion.
\end{proof}
Substituting the Taylor series Eq.~(\ref{Sdiag})  into Eq.~(\ref{Heff}) one
arrives at Eq.~(\ref{Hdiag}).
\end{proof}

\subsection{Convergence radius}
\label{subs:convergence}
In this section we analyze convergence of the formal series
$S=\sum_{j=1}^\infty S_j \, \epsilon^j$ and $H_{\eff}=\sum_{j=1}^\infty H_{\eff,j}\,\epsilon^j$.
\begin{lemma}
The series for $S$ and $H_{\eff}$ converge absolutely  in the disk $|\epsilon|<\rho_c$ where
\[
\rho_c=\frac{\epsilon_c}{8\left(1+ \frac{2|\calI_0|}{\pi \Delta}\right)}.
\]
Here $\epsilon_c=\Delta/(2\|V\|)$ and $|\calI_0|$ is the width of the interval $\calI_0$.
\end{lemma}
\begin{figure}
\centerline{
\includegraphics[height=3cm]{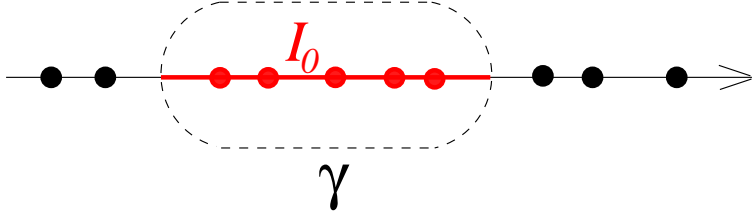}}
\caption{The contour $\gamma$ (dashed black curve) encircling the interval $\calI_0$ (bold red). Eigenvalues of
$H_0$ are indicates by solid dots.}
\label{fig:contour}
\end{figure}
\begin{proof}
Define projectors $Q_0=I-P_0$ and $P=I-P$.  Simple algebra shows that
\[
R_{P_0} R_{P} = I+Z, \quad \mbox{where} \quad Z=-2(P_0 Q + Q_0 P).
\]
For any $\rho>0$ define a disk $D(\rho)=\{ \epsilon \in \CC\, : \, |\epsilon|<\rho\}$.
It is well-known that  $P=P(\epsilon)$ has analytic continuation to $D(\epsilon_c)$.
Indeed, choose a contour $\gamma$ in the complex plane as the boundary
of the $(\Delta/2)$-neighborhood of the interval $\calI_0$, see Fig.~\ref{fig:contour}.
Then any eigenvalue of $H_0$ has distance at least $\Delta/2$ from $\gamma$
and $\| (zI-H_0)^{-1}\| \le 2/\Delta$ for all $z\in \gamma$. It implies that for any $z\in \gamma$
the resolvent $(zI-H)^{-1}$ is analytic in $D(\epsilon_c)$.
Thus we can define
\be
\label{Panl}
P=\frac1{2\pi i}\, \int_\gamma dz (zI - H_0-\epsilon V)^{-1}
\ee
and the Taylor series $P=P_0+\sum_{j=1}^\infty P_j \, \epsilon^j$ converges absolutely in the disk $D(\epsilon_c)$.
Furthermore, one can easily check that $P$ is a projector, $P^2=P$, for all $\epsilon\in D(\epsilon_c)$.
Note however that $P$ is not hermitian for complex $\epsilon$ and thus one may have $\|P\|>1$
(clearly,  $\|P\|\ge 1$ for any non-zero projector $P$ since $\la \psi|P|\psi\ra=1$ for some state $\psi$).

Using the standard perturbative expansion of the resolvent in Eq.~(\ref{Panl}) one gets a Taylor series
\[
P=P_0 + \sum_{j=1}^\infty P_j \,\epsilon^j,\quad P_j =  \frac1{2\pi i}\, \int_\gamma dz (zI-H_0)^{-1} \left( V (zI-H_0)^{-1} \right)^j.
\]
Noting that the contour $\gamma$ has length  $|\gamma| = \pi \Delta + 2|\calI_0|$, we can bound the norm of $P_j$ for $j\ge 1$ as
\[
\|P_j\, \epsilon^j\| \le \left( 1+\frac{2|\calI_0|}{\pi \Delta}\right) \cdot \left(\frac{2\|\epsilon V\|}{\Delta}\right)^j
= \left( 1+\frac{2|\calI_0|}{\pi \Delta}\right) \cdot (|\epsilon|/\epsilon_c)^j.
\]
It follows that
\[
\|P_0 Q\|=\| P_0 - P_0 P\| \le \sum_{j=1}^\infty \|P_0 P_j\, \epsilon^j\|
\le \sum_{j=1}^\infty \|P_j\, \epsilon^j\|
  \le  \left( 1+\frac{2|\calI_0|}{\pi \Delta}\right)
\cdot \frac{|\epsilon|}{\epsilon_c - |\epsilon|}<\frac14
\]
where the last inequality holds for $|\epsilon|<\rho_c$. The same bound applies
to  $\|Q_0 P\|$.
It follows that $Z$ is analytic in the disk $D(\rho_c)$ and $\|Z\|\le 2(\|P_0 Q\| + \|Q_0 P\|)<1$.
Thus the series for $S=\log{(U)}=(1/2)\log{(I+Z)}$
converges absolutely  in the disk $|\epsilon|<\rho_c$.
\end{proof}

\subsection{Additivity of the effective Hamiltonian}
\label{subs:add}
Consider a bipartite system of Alice and Bob with a Hilbert space $\calH=\calH^A \otimes \calH^B$.
Assume that both $H_0$ and $V$ can be represented as a sum of local terms acting
non-trivially only on Alice or Bob,
\[
H_0=H_0^A\otimes I^B + I^A\otimes H_0^B \quad \mbox{and} \quad V=V^A\otimes I^B + I^A \otimes V^B.
\]
In this section we choose the low-energy subspaces $\calP^A_0\subseteq \calH^A$
and $\calP^B_0\subseteq \calH^B$ as the ground state subspaces of Alice's
and Bob's Hamiltonians $H^A_0$ and $H^B_0$ respectively.
Let $\calP^A \subseteq \calH^A$ and $\calP^B\subseteq \calH^B$ be the ground state
subspaces of the perturbed Alice's and Bob's Hamiltonians $H_0^A + \epsilon V^A$
and $H_0^B + \epsilon V^B$ respectively.
Let $U^A$ and $U^B$ be the local SW-transformations on Alice's and Bob's subsystems such that
\be
\label{papb}
U^A P^A (U^A)^\dag = P^A_0 \quad \mbox{and} \quad U^B P^B (U^B)^\dag =P^B_0.
\ee
Applying the SW-transformation independently to Alice's and Bob's subsystems one obtains a pair
of effective Hamiltonians
\[
H_{\eff}^A=P^A_0 U^A (H_0^A + \epsilon V^A)(U^A)^\dag P^A_0 \quad
\mbox{and} \quad
H_{\eff}^B=P^B_0 U^B (H_0^B + \epsilon V^B)(U^B)^\dag P^B_0.
\]
Since there are no interactions between the two subsystems, the combined low-energy
effective Hamiltonian can be defined as $H_{\eff}^A\otimes I^B + I^A \otimes H_{\eff}^B$.

On the other hand, one can consider  the global SW-transformation $U^{AB}$ acting the full Hilbert space
$\calH=\calH^A\otimes \calH^B$ and mapping the ground subspace of $H_0+\epsilon V$
to the ground subspace of $H_0$ such that
\be
\label{pab}
U^{AB} (P^A \otimes P^B) (U^{AB})^\dag =P^A_0\otimes P^B_0.
\ee
It generates the effective low-energy Hamiltonian
\[
H_{\eff}^{AB}=(P^A_0\otimes P^B_0) U^{AB} (H_0+\epsilon V) (U^{AB})^\dag  (P^A_0\otimes P^B_0).
\]
As was pointed out in Section~\ref{subs:mult}, in general $U^{AB}\ne U^A \otimes U^B$.
Nevertheless it turns out that the effective low-energy Hamiltonians
obtained `globally' and `locally' are the same. We shall refer to this property as
{\em additivity} of the effective Hamiltonian. The additivity will play the key role in Section~\ref{sec:manybody}
in which we prove that the effective low-energy Hamiltonians describing weakly interacting
spin systems inherit locality properties of the original Hamiltonian (the so called linked cluster property).
\begin{lemma}[\bf Additivity]
\label{lemma:add}
Let $H_{\eff}^A$, $H_{\eff}^B$, and $H_{\eff}^{AB}$ be the low-energy effective Hamiltonians
defined above. Then
\be
\label{additivity0}
H_{\eff}^{AB}=H_{\eff}^A \otimes I^B + I^A \otimes H_{\eff}^B.
\ee
\end{lemma}
\begin{proof}
The additivity of $H_{\eff}$ follows directly from the weak multiplicativity of the direct rotation, see
Lemma~\ref{lemma:mult}. Indeed, using Eq.~(\ref{mult}), Eq.~(\ref{papb}), and Eq.~(\ref{pab}) one gets
\[
(P^A_0 \otimes P^B_0) U^{AB} = (P^A_0 \otimes P^B_0) (U^A \otimes U^B)
\quad \mbox{and} \quad
(U^{AB})^\dag (P^A_0 \otimes P^B_0)= (U^A \otimes U^B)^\dag (P^A_0 \otimes P^B_0)
\]
which implies Eq.~(\ref{additivity0}).
\end{proof}
It is worth mentioning that additivity of $H_{\eff}^{AB}$ does not apply to individual diagrams in
its perturbative expansion, see Section~\ref{subs:diagram}.
The simplest counter-example is a system of two qutrits (spins $S=1$), that is,
$\calH^A=\calH^B=\CC^3$. Choosing a basis of $\CC^3$ as $\{ |0\ra, |1\ra, |2\ra \}$,
we can define Hamiltonians $H^A_0=H^B_0=|2\ra\la 2|$ such that
$\calP_0^A=\calP_0^B=\CC^2$ is spanned by $|0\ra$ and $|1\ra$.
Choosing $V^A$ and $V^B$ as generic (random) $3\times 3$ hermitian matrices
we computed numerically individual diagrams $O(T_1), O(T_2), O(T_3)$
that contribute to
the fourth-order effective Hamiltonian $H_{eff,4}^{AB}$, see Eqs.~(\ref{T123}).
Note that $O(T_i)$ are two-qubit operators acting on
$\calP^A_0\otimes \calP^B_0 =\CC^2\otimes \CC^2$.
We observed numerically that $\la 0,0| O(T_i)|1,1\ra\ne 0$
for all $i=1,2,3$.
It shows that individual diagrams $O(T_i)$ contain interactions between $A$ and $B$.

\section{Schrieffer-Wolff theory for quantum many-body systems}
\label{sec:manybody}

In this section we specialize the SW method to many-body systems such as spin chains and spin
lattices. The Hilbert space describing these systems has a  a tensor
product structure $\calH=\bigotimes_{a=1}^N \calH_a$, where $\calH_a$ is a local Hilbert
space describing the $a$-th spin and $N$ is the total number of spins.
 We shall assume that the unperturbed Hamiltonian $H_0$ is a sum
of single-spin operators, so that its low-energy subspace $\calP_0$ by itself  has a tensor product structure,
$\calP_0=\bigotimes_{a=1}^N \calP_{0,a}$, where $\calP_{0,a}\subseteq \calH_a$ is the low-energy subspace
of the $a$-th spin. A perturbation $V$ will be chosen as a sum of two-spin interactions
such that each spin can interact only with $O(1)$ other spins.
Our assumption on $H_0$ and $V$ are formally stated in Section~\ref{subs:defs}.

The main technical difficulty with applying general  techniques of Section~\ref{sec:SW}
to many-body systems such as spin lattices is that the norm $\|\epsilon V\|$
is typically a macroscopic quantity proportional to the total number of spins $N$. On the other
hand, the spectral gap $\Delta$ of the unperturbed Hamiltonian $H_0$ does not grow with $N$.
It means that $\|\epsilon V\|\sim N \gg \Delta$ and thus the SW transformation is not well-defined,
see Lemma~\ref{lemma:per1}. By the same token, the Taylor series for the SW generator $S$ and the effective Hamiltonian $H_{\eff}$ derived in Section~\ref{subs:series} could be divergent. We shall get around this difficulty by
considering the Taylor series $H_{\eff}=\sum_{q\ge 0} H_{\eff,q} \, \epsilon^q$ truncated at some fixed order $n$,
\be
\label{Htruncated}
H_{\eff}^{\la n\ra}  = P_0 H_0  + \epsilon P_0 V P_0 + \sum_{q=2}^n H_{\eff,q} \, \epsilon ^q.
\ee
Note that $H_{\eff}^{\la n\ra}$ is a polynomial in $\epsilon$ of finite degree, so it is well defined even if the
series for $H_{\eff}$ is divergent.

Since the low-energy subspace $\calP_0$ has a tensor product structure, one can ask
whether the truncated Hamiltonian $H_{\eff}^{\la n\ra}$ can be written as a sum of
local interactions. In Section~\ref{subs:lct} we use additivity of $H_{\eff}$
to show that a Taylor coefficient $H_{\eff,q}$ includes only interactions among subsets of spins that are spanned by
{\em connected} clusters of $q$ edges in the interaction graph, see Theorem~\ref{thm:lct}.
It demonstrates that the SW transformation maps a high-energy Hamiltonian with local interactions to
an effective low-energy Hamiltonian with (approximately) local interactions.

Our main result is the following theorem.
\begin{theorem}
\label{thm:global}
There exists a constant threshold $\epsilon_c>0$ depending only on $n$, $\Delta$, and the maximum
degree of the interaction graph such that for all $|\epsilon|<\epsilon_c$ the ground state
energy of $H_{\eff}^{\la n\ra}$ approximates the ground state energy of $H_0+\epsilon V$
with an error at most $\delta_n=O(N|\epsilon|^{n+1})$.
\end{theorem}

We begin by proving the theorem for the special case of so-called block-diagonal perturbations,
see Section~\ref{subs:bd}. These are perturbations composed of local interactions preserving the
low-energy subspace $\calP_0$. For block-diagonal perturbations the SW transformation is trivial,
that is, $S=0$, see Section~\ref{subs:series}. It follows that
$H_{\eff}=P_0 H_0 + \epsilon P_0 V P_0$. Clearly, $H_{\eff}$ can reproduce ground state properties
of the exact Hamiltonian $H_0+\epsilon V$ only if at least one ground state of $H_0+\epsilon V$
belongs to the low-energy subspace $\calP_0$.
We prove that this is indeed the case if $|\epsilon|<\epsilon_c'$ for some
constant $\epsilon_c'>0$ that depends only on the gap $\Delta$, maximum degree of the interaction
graph, and the strength of interactions in $V$, see Lemma~\ref{lemma:new}.

The role of the SW transformation is to map generic perturbations to block-diagonal perturbations.
Unfortunately, the SW transformation cannot
be used directly since the transformed truncated Hamiltonian $U(H_0+\epsilon V)U^\dag$
becomes local only when restricted to  $\calP_0$, while Lemma~\ref{lemma:new} requires locality on the
entire Hilbert space. In addition, the Taylor coefficients $S_p$ of the SW generator
are highly non-local operators which complicates an  analysis of the error obtained by truncating
the series at some finite order.
We resolve this problem by employing an auxiliary perturbative expansion
due to Datta et al~\cite{DFFR} which we refer to as a {\em local Schrieffer-Wolff} transformation,
see Section~\ref{subs:DFFR}.
It defines an effective low-energy Hamiltonian
\[
H_{\eff,\loc}=P_0\,  e^T (H_0+\epsilon V)e^{-T} P_0
\]
for some unitary transformation $e^T$.
This transformation is constructed such that
the transformed Hamiltonian $e^T (H_0+\epsilon V)e^{-T}$ is
block-diagonal, that is, it can be written as a sum of (approximately) local interactions
preserving $\calP_0$. Here the locality of interactions holds on the full Hilbert space
which allows us to apply Lemma~\ref{lemma:new}.
 In Section~\ref{subs:DFFR} we prove
an analogue of Theorem~\ref{thm:global} for the $n$-th order effective Hamiltonian constructed using the local SW method.
The proof uses some of techniques developed by Datta et al~\cite{DFFR} to bound the error
resulting from truncating the series for the generator $T$.

Finally, we construct  a unitary transformation $e^K$ acting on the low-energy subspace $P_0$
that maps the effective low-energy
Hamiltonians obtained using the standard and the local SW methods to each other,
that is, $H_{\eff}=e^K H_{\eff,\loc} e^{-K}$, see Section~\ref{subs:equiv}.
The relationship between low-energy theories
obtained using the standard and the local SW methods is illustrated on Fig.~\ref{fig:iso}.
We use weak multiplicativity of the SW transformation proved in Section~\ref{subs:mult}
to show that the perturbative expansion of $K$ contains only linked clusters.
This linked cluster property of $K$ allows us to bound the error resulting from truncating
the series for $K$ at the $n$-th order by $O(N|\epsilon|^{n+1})$.
Combining the two errors we obtain the desired bound on $\delta_n$ in Theorem~\ref{thm:global}.

\subsection{The unperturbed Hamiltonian and the perturbation}
\label{subs:defs}
Let $\Lambda$ be a lattice or a graph with $N$ sites such that each site $u\in \Lambda$ is occupied by
a finite-dimensional particle (spin) with a local Hilbert space $\calH_u$. Accordingly, the full Hilbert space is
\[
\calH=\bigotimes_{u\in \Lambda} \calH_u.
\]
We will choose  the unperturbed Hamiltonian $H_0$ as a sum of single-spin operators,
\be
H_0=\sum_{u\in \Lambda} H_{0,u},
\ee
where $H_{0,u}$ acts non-trivially only on a spin $u$.
By performing an overall energy shift we can always assume that each term
$H_{0,u}$ is a positive semi-definite operator and the ground state energy of $H_{0,u}$ is
zero. Let $\calP_{0,u} \subseteq \calH_u$ be the ground  subspace of a spin $u$
and $P_{0,u}$ be the tensor product of the projector onto $\calP_{0,u}$ and the identity operator
on all other spins.
Then the ground subspace of the entire Hamiltonian $H_0$ is
\[
\calP_0=\bigotimes_{u\in \Lambda} \calP_{0,u}
\]
and the projector onto $\calP_0$ can be represented  as
\[
P_0=\prod_{u\in  \Lambda} P_{0,u}.
\]
Let $Q_0=I-P_0$ be the projector onto $\calP_0^\perp$ and $Q_{0,u}= I-P_{0,u}$.
We shall assume that each term $H_{0,u}$ has a spectral gap at least $\Delta$ above the ground state.
Thus the full Hamiltonian $H_0$ also has a spectral gap at least $\Delta$.
In general, the maximal norm of the local terms $\|H_{0,u}\|$ is a parameter independent of $\Delta$.
However, for the sake of simplicity we shall assume that
\[
\|H_{0,u}\|=O(\Delta) \quad \mbox{for all $u\in \Lambda$}.
\]

It is worth mentioning that if we allow each spin to have two or more distinct levels in the low-energy subspace
$\calP_{0,u}$, the full Hamiltonian might not have a spectral gap between the subspaces $\calP_0$ and $\calP_0^\perp$.
Consider as an example the case when each spin has $3$ basis states $|0\ra, |1\ra,|2\ra$ which have energy
$0$, $\delta>0$, and $\Delta\gg \delta$ respectively such that the low-energy subspace $\calP_{0,u}$ is spanned by the states
$|0\ra$ and $|1\ra$. Then the state $|1^{\otimes N}\ra\in \calP_0$ has energy $N\delta$ while the state
$|2,0^{\otimes (N-1)}\ra\in \calP_0^\perp$ has energy $\Delta$. Thus for any constant $\delta$ and $\Delta$ the spectral
gap between $\calP_0$ and $\calP_0^\perp$ closes if $N\ge \Delta \delta^{-1}$.

We will choose the perturbation $V$ as a sum of two-spin interactions between nearest neighbor sites,
\[
V=\sum_{(u,v)\in \calE} V_{u,v}.
\]
Here $\calE$ is the set of edges of the underlying lattice (which could be an arbitrary graph)
and $V_{u,v}$ is an arbitrary hermitian operator acting on a pair of sites $u,v$.
To describe the effective low-energy theory we
will have to consider {\em $k$-local} perturbations,
that is, Hamiltonians with at most $k$-spin interactions. A $k$-local perturbation $V$
can be written as
\[
V=\sum_{A\subseteq \Lambda} V_A,  \quad \quad \mbox{$V_A=0$ unless $|A|\le k$}.
\]
Here $V_A$ is an arbitrary hermitian operator acting on a subset of spins $A$.
Define a {\em strength}  of the perturbation as
\be
\label{norm}
J=\max_{u\in \Lambda} \; \sum_{A\subseteq \Lambda\, : \, A\ni u} \|V_A\|.
\ee
The strength of the perturbation provides an upper bound on the combined norm of all interactions affecting any selected spin.
If $\Lambda$ is a regular lattice in a space of bounded dimension then $J$
coincides up to a constant factor with the maximum norm of interactions, $J\sim \max_A \|V_A\|$.
We will always assume that $J$ and $\Delta$ are constants independent of $N$.
Let us say that a perturbation $V=\sum_{A\subseteq \Lambda} V_A$ is {\em block-diagonal}
iff $V_A$ preserves the low-energy subspace $\calP_0$, or equivalently, $[V_A,P_0]=0$.

\subsection{Block-diagonal perturbations}
\label{subs:bd}

Recall that the goal of the SW transformation $U$ is to transform the perturbed Hamiltonian
$H_0+\epsilon V$ into a block-diagonal form such that the transformed Hamiltonian
$U(H_0+\epsilon V)U^\dag$ preserves the low-energy subspace $\calP_0$.
In this section we analyze the special case when the perturbed Hamiltonian $H_0+\epsilon V$ is already
block-diagonal, so that $U=I$. In this case the effective low-energy Hamiltonian is simply
$H_{\eff}=P_0 (H_0+\epsilon V)P_0$. Since we assumed that $H_0 P_0=0$, one has
$H_{\eff}=\epsilon P_0 V P_0$. A natural question is whether the ground state energy of $H_{\eff}$
coincides with the ground state energy of $H_0+\epsilon V$. Clearly, this happens iff at least one
 ground state of $H_0+\epsilon V$
belongs to the subspace $\calP_0$. In this section we prove that this is indeed the case provided that the strength
of the perturbation $\epsilon V$ is below certain constant threshold depending only on $\Delta$ and $k$.
Intuitively it follows from the fact that exciting any subset of $w$ spins increases the energy of the term $H_0$
by at least $w\Delta$, while the energy of the term $\epsilon V$ can go down at most by $2w |\epsilon| J$. Hence exciting
spins can only increase the energy provided that $|\epsilon|<\epsilon_c\sim \Delta/J$.
Unfortunately this simple argument does not tell anything about states that contain superpositions of different excited subsets
with different $w$, and  one may expect that exciting subsets of spins ``in a superposition" can
help to reduce the overall energy. We analyze this general situation in the following lemma
(to simplify notations we include the factor $\epsilon$ into the definition of $V$).
\begin{lemma}
\label{lemma:new}
Let $V=\sum_{A\subseteq \Lambda} V_A$ be a $k$-local perturbation with strength $J$.
Suppose that each term $V_A$ is block-diagonal, that is,
$V_A$ preserves the subspace $\calP_0$.  Assume also that  $2^{k+2} J <\Delta$.
Then the subspace $\calP_0$ contains at least one ground state of $H_0+V$.
\end{lemma}
Let us emphasize that the lemma does not tell anything about the spectral gap of the perturbed Hamiltonian
$H_0+V$ separating the subspaces $\calP_0$ and $\calP_0^\perp$.
We believe that this gap can be closed by perturbations with a strength of order $1/N$.
\begin{proof}
Let us assume that $H_0+V$ has a ground state $|\psi\ra\in \calP_0^\perp$. We will show that this assumption leads to a contradiction. Indeed, we can represent such a ground state $|\psi\ra$ as  a sum over $2^N-1$ configurations of excited  spins,
\be
\label{sumx}
|\psi\ra=\sum_{x\in \{0,1\}^N} |\psi_x\ra,
\ee
where $x$ is a binary string of length $N$ such that
$x_u=1$ indicates that a spin $u$ is excited ($P_{0,u}\, |\psi_x\ra = 0$) and $x_u=0$
indicates that a spin $u$ is not excited ($P_{0,u}\, |\psi_x\ra=|\psi_x\ra$).
More formally,
\[
|\psi_x\ra = R(x)\, |\psi\ra, \quad \mbox{where} \quad R(x)=\prod_{u\in \Lambda} ((1-x_u) P_{0,u} + x_u Q_{0,u}).
\]
We note that $|\psi_x\ra$ may be a highly entangled state.
Since we assumed $|\psi\ra\in \calP^\perp$, the configuration with no excited spins does not appear
in $|\psi\ra$, that is,
\[
|\psi_{0^N}\ra=0.
\]
Let us  choose a configuration $x^*\in \{0,1\}^N$  with the largest amplitude
\[
\alpha_x=\la \psi_x|\psi_x\ra,
\]
that is,
\be
\label{alpha_max}
\alpha_{x^*}\ge \alpha_x \quad \mbox{for all $x$}.
\ee
The key idea is to transform $|\psi\ra$ to  a new state $\eta$ (which is generally mixed)
by  ``annihilating" all excitations in $|\psi_{x^*}\ra$ while leaving all other components
$|\psi_x\ra$, $x\ne x^*$, untouched.
We will show that for sufficiently small strength of $V$ the energy of  $\eta$ is smaller than the
energy of $|\psi\ra$ thus arriving at a contradiction.

For any spin $u\in \Lambda$ choose an arbitrary ground state $|\omega_u\ra \in \calP_{0,u}$.
Define a trace preserving completely positive (TPCP) map $\calE_u$
that erases the state of the spin $u$ and replaces it by the chosen  ground state $|\omega_u\ra$.
More formally, for any (mixed) state $\tau$ describing the entire system one has
\[
\calE_u(\tau) =|\omega_u\ra\la \omega_u| \otimes \trace_u (\tau).
\]
Note that $\calE_u(\tau)$ may be mixed even if $\tau$ is a pure state.
The map $\calE_u$ can be regarded as an ``annihilation superoperator".
Obviously, annihilation superoperators on different spins commute with each other.
Let $\rho$ be a state obtained from $|\psi_{x^*}\ra$
by applying the annihilation superoperator to every excited spin,  that is,
\[
\rho=\left( \prod_{u\, : \, x^*_u=1} \calE_u\right) (|\psi_{x^*}\ra\la \psi_{x^*}|).
\]
Clearly, $\calE_u$ reduces the energy of the term $H_{0,u}$ at least by
$\Delta \alpha_{x^*}$ for any excited spin $u$. Thus if we denote the
number of excited spins (the Hamming weight of the string $x^*$) by $w$ then
\be
\label{reduce1}
\trace{(\rho \, H_0)} \le \la \psi_{x^*}|H_0|\psi_{x^*}\ra  -w\Delta\, \alpha_{x^*}.
\ee
Now consider a state obtained from $|\psi\ra$ be removing the largest-amplitude term,
\[
|\psi_{\els}\ra=\sum_{x\ne x^*} |\psi_x\ra.
\]
Denote
\be
\label{new_state}
\eta=\rho+|\psi_{\els}\ra\la \psi_{\els}|.
\ee
Since $\calE_u$ are trace preserving maps, we have $\trace\rho =\la \psi_{x^*}|\psi_{x^*}\ra$
and thus $\trace\eta = \la \psi|\psi\ra=1$.
Let us show that for sufficiently small strength of $V$ the energy of $\eta$ is smaller than the
energy of $|\psi\ra$, thus getting a contradiction. Indeed, since $H_0$ cannot change
a configuration of excited spins we have
\be
\label{reduction1}
\trace{(\eta H_0)} - \la \psi|H_0|\psi\ra = \trace{(\rho \, H_0)} -\la \psi_{x^*}|H_0|\psi_{x^*}\ra \le -w\Delta\, \alpha_{x^*},
\ee
see Eq.~(\ref{reduce1}). It remains to bound the contribution to the energy
coming from the perturbation $V$.  Simple algebra leads to
\be
\label{Wcon}
\trace{(\eta V)} -\la \psi|V|\psi\ra =\sum_{A\subseteq \Lambda}
\trace{(\rho V_A)} - \la \psi_{x^*} |V_A|\psi_{x^*}\ra - 2\mbox{Re}(\la \psi_{x^*}|V_A|\psi_{else}\ra).
\ee
Note that  the states $\rho$ and $|\psi_{x^*}\ra$ have the same reduced density matrices on any subset
$A\subseteq \Lambda$  that contains no excited spins, that is, $x^*_u=0$ for all $u\in A$.
It follows from the fact that annihilation superoperators applied outside $A$ cannot change
the reduced state of $A$.
Therefore
\be
\label{zero_con1}
\trace{(\rho V_A)} -\la \psi_{x^*} |V_A|\psi_{x^*}\ra =0 \quad \mbox{if $x^*_u=0$ for all $u\in A$}.
\ee
Taking into account that any interaction $V_A$ is block-diagonal, the operator $V_A$
can change a configuration of excited spins only if at least one spin in $A$ is already excited.
It means that
\be
\label{zero_con2}
\la \psi_{x^*}|V_A|\psi_{else}\ra=0 \quad \mbox{if $x^*_u=0$ for all $u\in A$}.
\ee
Using Eq.~(\ref{zero_con1}) we get a bound
\be
\label{Wbound1}
\left| \sum_{A\subseteq \Lambda} \trace{(\rho V_A)} - \la \psi_{x^*} |V_A|\psi_{x^*}\ra  \right|
\le
\sum_{u\, : \, x^*_u=1} \sum_{A\ni u}
\left| \trace{(\rho V_A)} - \la \psi_{x^*} |V_A|\psi_{x^*}\ra  \right|\le 2w\alpha_{x^*} J.
\ee
It bounds the contribution to the energy from the first and the second terms in Eq.~(\ref{Wcon}).
Now let us bound the contribution from the last term in Eq.~(\ref{Wcon}).  Using Eq.~(\ref{zero_con2})
we get
\bea
\label{Wbound2}
\left| \sum_{A\subseteq \Lambda}
2\mbox{Re}(\la \psi_{x^*}|V_A|\psi_{else}\ra) \right|  &\le&
\sum_{u\, : \, x^*_u=1}\,  \sum_{A\ni u} \, \sum_{x\ne x^*} 2|\la \psi_{x^*} |V_A|\psi_x\ra| \nn \\
&\le &
\sum_{u\, : \, x^*_u=1}\,  \sum_{A\ni u}  2^{k+1} \|V_A\| \sqrt{\alpha_{x^*} \alpha_x} \nn \\
&\le & 2\cdot 2^k w \alpha_{x^*}  J.
\eea
Here the first line follows from Eq.~(\ref{zero_con2}) and definition of $|\psi_{else}\ra$.
The second line follows from  the Cauchy-Schwarz inequality and the fact that $V_A$ acts
non-trivially on at most $k$ spins, so the number of possible transitions $x^*\to x$ is upper
bounded by $2^k$. The third line follows from the
maximality condition Eq.~(\ref{alpha_max}).
Combining bounds Eq.~(\ref{Wbound1},\ref{Wbound2}) we arrive at
\be
\label{Wbound}
\left| \trace{(\eta V)} - \la \psi|V|\psi\ra \right| \le 2^{k+2} w\alpha_{x^*} J.
\ee
Combining it with Eq.~(\ref{reduction1}) we finally get
\be
\trace{(\eta H)} - \la \psi|H|\psi\ra \le (-\Delta + 2^{k+2} J) w\alpha_{x^*}.
\ee
Since by assumption $x^*$ contains at least one excitation, we have $w>0$.
It means that the energy of $\eta$ is strictly smaller than the one of $|\psi\ra$
if $2^{k+2} J<\Delta$. However this is impossible since $|\psi\ra$ is a ground state.
Thus our assumption $|\psi\ra\in \calP^\perp$ leads to a contradiction.
\end{proof}
 Lemma~\ref{lemma:new}
can be easily generalized to the case when $V$ contains interactions among
arbitrarily large number of spins, but the magnitude of $k$-spin interactions
decays exponentially with $k$.
\begin{corol}
\label{corol:overall}
Consider a Hamiltonian $V=\sum_{a=1}^n  V_a$
such that $V_a$ is a $k(a)$-local perturbation with strength $J(a)$.
Suppose that all interactions in each perturbation $V_a$ are block-diagonal.
Assume also that
\be
\label{overall}
\sum_{a=1}^n 2^{k(a)+2}\, J(a) <\Delta.
\ee
Then the subspace $\calP_0$ contains at least one ground state  of $H_0+V$.
\end{corol}

\subsection{Linked cluster theorem}
\label{subs:lct}

Let $\epsilon_{u,v}$ be formal variables associated with the edges of the lattice $(u,v)\in \calE$.
Consider a perturbation
\[
V=\sum_{(u,v)\in \calE} \epsilon_{u,v} V_{u,v}.
\]
Using the perturbative expansion described in Section~\ref{subs:series}
one can represent the effective low-energy Hamiltonian $H_{\eff}$ acting on $\calP_0$
as a formal operator-valued multivariate series in variables $\epsilon_{u,v}$.
Clearly any Taylor coefficient $H_{\eff,n}$ is a homogeneous multivariate polynomial
in $\epsilon_{u,v}$ of degree $n$. One can uniquely represent $H_{\eff,n}$ as a sum of monomials
\[
K \, \prod_{(u,v)\in C} \epsilon_{u,v}^{m_{u,v}}, \quad m_{u,v}\ge 1, \quad \sum_{(u,v)\in C} m_{u,v}=n,
\]
where $K$ is some operator acting on $\calP_0$ and $C\subseteq \calE$ is a subset of edges of the lattice.
We shall refer to $C$ as the {\em support} of the above monomial. Thus we have
\be
\label{monomials}
H_{\eff,n}=\sum_{C\subseteq \calE} K_{n,C},
\ee
where $K_{n,C}$ is defined as the sum of all monomials
in $H_{\eff}$ that have support $C$ and degree $n$.
Given any set of edges $C\subseteq \calE$ let
$\Lambda(C)\subseteq \Lambda$ be the corresponding set of sites,
such that $u\in \Lambda(C)$ iff $C$ contains an edge incident to $u$.

Let us point out that the only purpose of introducing the variables $\epsilon_{u,v}$ is to
define the operators $K_{n,C}$ and  the decomposition Eq.~(\ref{monomials}).
Once this decomposition is defined, we can set $\epsilon_{u,v}=\epsilon$
for all edges.  Also note that we do not make any assumptions about convergence of $H_{\eff}$ considered as
a multivariate series\footnote{Let us remark that
absolute convergence of the univariate series $H_{\eff}=\sum_{n=1}^\infty H_{\eff,n} \epsilon^n$ does not imply absolute convergence
of the corresponding multivariate series, since different terms in the multivariate series may cancel out each other when combined into a single term $H_{\eff,n}$.}
since we only consider some fixed order $n$.
The main result of this section is the following theorem.
\begin{theorem}[\bf linked cluster theorem]
\label{thm:lct}
The operator $K_{n,C}$ acts only on the subset of spins $\Lambda(C)$.
In addition, $K_{n,C}=0$ unless $C$ is a connected subset of edges and $|C|\le n$.
\end{theorem}
\noindent
Thus the $n$-th order effective Hamiltonian includes only interactions among ``linked clusters''
of spins of size at most $n$. However let us point out that the theorem does not say that $K_{n,C}$ acts non-trivially
on {\em all} spins in $\Lambda(C)$.
\begin{proof}
Choose any subset of edges  $C\subseteq \calE$ and set $\epsilon_{u,v}=0$ for all $(u,v)\notin C$.
Partition the lattice as $\Lambda=AB$  where $A=\Lambda(C)$
and $B$ is the complement of $A$. Then we have $V=V^A\otimes I^B$ for some operator $V^A$
acting on $A$. Since the unperturbed Hamiltonian $H_0$ does not include any interactions, additivity
of the effective Hamiltonian implies $H_{\eff,n}=H_{\eff,n}^A \otimes I^B$, where
$H_{\eff,n}^A$ is the effective low-energy Hamiltonian computed using only the subsystem $A$
with the perturbation $V^A$,  see Lemma~\ref{lemma:add} in Section~\ref{subs:add}.
We conclude that all monomials in the decomposition of $H_{\eff,n}$ act trivially on $B$.
By definition,  $K_{n,C}$ does not depend on the variables $\epsilon_{u,v}$
with $(u,v)\notin C$, so we infer that $K_{n,C}$ acts trivially on $B$.

Suppose now that $C$ consists of two disconnected subsets of edges $C_1$ and $C_2$.
Choose $A=\Lambda(C_1)$ and $B=\Lambda(C_2)$. Again we set $\epsilon_{u,v}=0$ for all $(u,v)\notin C$.
Then we have $V=V^A\otimes I^B + I^A \otimes V^B$, where $V^A\otimes I^B$ is obtained from $V$
by setting $\epsilon_{u,v}=0$ for  all $(u,v)\in C_2$.
Similarly, $I^A\otimes V^B$ is obtained from $V$
by setting $\epsilon_{u,v}=0$ for all $(u,v)\in C_1$.
Additivity of the effective Hamiltonian implies $H_{\eff,n}=H_{\eff,n}^A \otimes I^B + I^A \otimes H_{\eff,n}^B$,
where $H_{\eff,n}^A$ and $H_{\eff,n}^B$ are the effective Hamiltonians computed using only the subsystems
$A$ and $B$ with perturbations $V^A$ and $V^B$ respectively.
It means that no monomial in $H_{\eff,n}$ contains variables $\epsilon_{u,v}$ from both $C_1$ and $C_2$
simultaneously.
Since $K_{n,C}$ does not depend on the variables $\epsilon_{u,v}$
with $(u,v)\notin C$, we infer that $K_{n,C}=0$.

The statement $K_{n,C}=0$ for $|C|>n$ follows from the fact that $H_{\eff,n}$ contains only
monomials of degree $n$.
\end{proof}

It is obvious from the proof that the linked cluster theorem applies to any
operator-valued multivariate series that obeys the additivity property.

\subsection{Local Schrieffer-Wolff trasformation}
\label{subs:DFFR}

Applying  the SW transformation described in Section~\ref{sec:SW} to many-body systems
one must keep in mind that its generator $S$ is a highly non-local operator even in the lowest orders of the perturbative expansion.
This follows from the fact that $S$ is a block-off-diagonal operator, that is
$P_0 S P_0 =0$ and $Q_0 S Q_0 =0$, see Lemma~\ref{lemma:S}.
One can easily check that any block-off-diagonal operator must act non-trivially on every spin in the system.
Indeed, block-off-diagonality implies $S=P_0 S+ S P_0$. If one assumes that $S$ acts trivially on some spin $u$,
then  $Q_{0,u} S Q_{0,u}= S Q_{0,u} \ne 0$. On the other hand,
$Q_{0,u} S Q_{0,u}=Q_{0,u} (P_0S +SP_0) Q_{0,u} =0$ since $Q_{0,u} P_0=P_0 Q_{0,u}=0$.
The non-locality of $S$ is not a very serious drawback if the main object one is interested in is the
effective low-energy Hamiltonian $H_{\eff}$. Indeed, the linked cluster theorem proved in the previous section  implies that all low-order
corrections to $H_{\eff}$ are in fact local operators. On the other hand,
the non-locality of $S$ makes it more difficult to analyze properties of $H_{\eff}$.
In this section we get around this difficulty by employing  a local version of the SW transformation
developed by Datta et al~\cite{DFFR}. To avoid a confusion between the two types of the SW
transformation we shall refer to the direct rotation described in Sections~\ref{sec:crot},\ref{sec:SW}
as a {\em global SW transformation}.

The local SW transformation is a unitary operator $U$ that brings the perturbed Hamiltonian $H_0+\epsilon V$
into  an {\em approximately} block-diagonal form. Specifically, for any fixed integer $n\ge 0$ we shall choose
$U=e^T$ where $T$ is a degree-$n$ polynomial in $\epsilon$,
\[
T=\sum_{q=1}^n T_q \, \epsilon^q, \quad T_q^\dag=-T_q.
\]
This polynomial will be chosen such that
 the transformed Hamiltonian
$e^T (H_0+\epsilon V) e^{-T}$ is block-diagonal if truncated at the order $n$.
The definition of $T$ is similar to the perturbative definition of $S$, see Section~\ref{subs:series}.
We use the Taylor expansion of $\exp{(\hat{T})}$  to get
\be
\label{HtransfT}
e^T(H_0+\epsilon V)e^{-T} = H_0 + \sum_{j=1}^n \epsilon^j \,\left( [T_j,H_0] +V^{(j-1)}\right)
+  \sum_{j=n+1}^\infty \epsilon^j\, V^{(j-1)},
\ee
where $V^{(j)}$ is a certain polynomial function of $\hat{T}_1,\ldots,\hat{T}_{\min{(j,n)}}$ applied to $H_0$ or $V$.
The precise form of this polynomial will be given in Eq.~(\ref{V(j-1)T}).
Strictly speaking, $V^{(j)}$ should also carry a label $n$, but since $n$ is fixed throughout this section,
we shall only retain the label $j$.

The main difference between the local and global SW method is the inductive rule
used to  define the generators of the transformation.
In the global SW method the generators $S_j$ are fixed
by demanding  that $S_j$ is block-off-diagonal while
the $j$-th order contribution to the transformed Hamiltonian
 is block-diagonal, see Section~\ref{subs:series}.
In the local SW method the generators $T_j$ are fixed by demanding  that
$T_j$ is a sum  of local interactions  (the range of interactions may depend on $j$)
and each interaction is {\em locally}  block-off-diagonal on the subset of spins
it acts upon.   For each local term in the expansion of $V^{(j-1)}$ we shall introduce
the corresponding local term in $T_j$ chosen such that
the expansion of $[T_j,H_0]+V^{(j-1)}$ contains only  block-diagonal terms.

To define the generators $T_j$ formally let us introduce some notations.
For any subset of spins $A\subseteq \Lambda$  define projectors onto the local  low-energy  and high-energy subspaces
\be
P_{A}=\prod_{u\in A} P_{0,u} \quad \mbox{and} \quad  Q_{A}=I-P_{A}.
\ee
Introduce superoperators  $\calO_A$ and $\calD_A$ such that
\[
\calO_A(X)=P_{A} X\, Q_{A}+ Q_{A} X\,  P_{A} \quad \mbox{and} \quad
 \calD_A(X)=P_{A}  X\,  P_{A}  + Q_{A} X\,  Q_{A}.
\]
Clearly, $X=\calO_A(X) + \calD_A(X)$ for any operator $X$.
We shall also need a local version of the superoperator $\calL$, see Eq.~(\ref{L0inv}).
It can be chosen as
\[
\calL_{A}(X)=\int_0^\infty dt \, e^{-tH_{0,A}} Q_{A} X P_{A}   -\int_0^\infty dt \,
P_{A} X Q_{A} e^{-tH_{0,A}},
\]
where $H_{0,A}=\sum_{u\in A} H_{0,u}$. Taking into account that $\| e^{-tH_{0,A}} Q_{A} \| \le e^{-\Delta t}$
 one can easily check that the integral over $t$ converges and
\be
\label{1/gap}
\|\calL_{A}(X) \| \le \frac{\|X\|}{\Delta}.
\ee
In addition, for any operator $X$ that acts non-trivially only on spins in $A$ one has an identity
\be
\label{L_Aid}
\calL_A([H_0,X])=[H_0,\calL_A(X)]=\calO_A(X).
\ee
Now we are ready to define the generators $T_j$. We begin by representing $V^{(j)}$ as a sum of local interactions,
\[
V^{(j)}=\sum_{A\subseteq \Lambda} V^{(j)}_A.
\]
Here $V^{(j)}_A$ acts non-trivially only on $A$ (for example, one can use the expansion of
$V^{(j)}$ in the generalized Pauli basis). Then we choose
\be
\label{DFFR-method}
T_j=\sum_{A\subseteq \Lambda} \calL_A (V^{(j-1)}_A) \quad \mbox{for all $1\le j\le n$}.
\ee
Using the identity Eq.~(\ref{L_Aid}) one can easily check that for this choice of $T_j$
\[
[T_j,H_0] + V^{(j-1)} = \sum_{A\subseteq \Lambda} \calD_A(V^{(j-1)}_A) \quad \mbox{for all $1\le j\le n$}.
\]
Accordingly, the transformed Hamiltonian can be written as
\be
e^T(H_0+\epsilon V)e^{-T} = H^{\la n\ra} + H_{\garb},
\ee
where
\be
H^{\la n\ra} = H_0+ \sum_{j=1}^n \epsilon^j \sum_{A\subseteq \Lambda} \calD_A(V^{(j-1)}_A)
\ee
is a sum of local block-diagonal interactions, while
\be
H_{\garb}= \sum_{j=n+1}^\infty \epsilon^j\, V^{(j-1)}
\ee
includes all unwanted terms generated at the order $n+1$ and higher.
Finally, we define the $n$-th order effective low-energy Hamiltonian as a restriction
of $H^{\la n\ra}$ onto the subspace $\calP_0$,
\be
H^{\la n\ra}_{\eff}=P_0 H^{\la n\ra} P_0.
\ee
It is obvious from this construction that $H^{\la n\ra}_{\eff}$ is additive for perturbations
describing two non-interacting disjoint systems, see Lemma~\ref{lemma:add}.
It follows that the Taylor coefficients of $H^{\la n\ra}_{\eff}$ obey the linked cluster property of
Theorem~\ref{thm:lct}.

For the later use let us give an explicit formula for the polynomial $V^{(j)}$.
We have $V^{(0)}=V$ and
\be
\label{V(j-1)T}
V^{(j-1)} = \sum_{q=2}^j \frac1{q!} \; \; \sum_{\substack{1\le j_1,\ldots j_q\le n \\ \\  j_1+ \ldots+ j_q=j\\ }} \hat{T}_{j_1} \cdots \hat{T}_{j_q} (H_0) +
 \sum_{q=1}^{j-1} \frac1{q!} \; \; \sum_{\substack{1\le j_1,\ldots j_q\le n \\  \\  j_1+ \ldots+ j_q=j-1\\}} \hat{T}_{j_1} \cdots \hat{T}_{j_q} (V),
\ee
for all $j\ge 2$.
Note that convergence of all above series in the limit $n\to \infty$ does not matter since we always
consider some fixed order $n$. However, it can be shown~\cite{DFFR} that the $n\to \infty$ limits of the series
for $T$ and $H^{\la n\ra}_{\eff}$ converge absolutely in the disk of radius $\epsilon_c\sim 1/N^2$.

Below  we shall prove that the norm of $H_{\garb}$ can be bounded by $O(N|\epsilon|^{n+1})$
for sufficiently small $\epsilon$. Then we shall employ Lemma~\ref{lemma:new} and its Corollary~\ref{corol:overall}
to show that
the ground state of $H^{\la n\ra}$ belongs to the subspace $\calP_0$ for sufficiently small $\epsilon$.
Combining the two results together we shall get the following theorem which can be regarded as
an analogue of Theorem~4.7 in~\cite{DFFR}.
\begin{theorem}
\label{thm:local}
Let $H^{\la n\ra}_{\eff}$ be the $n$-th order low-energy effective Hamiltonian constructed using
the local SW method.
There exists a constant threshold
\be
\label{local_threshold}
\epsilon_c=\Omega(\Delta n^{-2})
\ee
such that for all $|\epsilon|<\epsilon_c$ the ground state energy of $H^{\la n\ra}_{\eff}$
approximates the ground state energy of $H_0+\epsilon V$ with an error $\delta_n=O(N|\epsilon|^{n+1})$.
\end{theorem}
Given any perturbation $V$ represented as a sum of local interactions, $V=\sum_{A\subseteq \Lambda} V_A$
we shall use the notation $\|V\|_1$ for the strength of $V$, that is,
\[
\| V\|_1 = \max_{u\in\Lambda} \; \sum_{A\ni u} \| V_A\|.
\]
All derivations below will implicitly use a bound
\be
\label{trivial_bound}
\sum_{A\, : \, A\cap M\ne \emptyset} \|V_A\| \le |M|\cdot \| V\|_1
\ee
that holds for any operator $V=\sum_{A\subseteq \Lambda} V_A$ and for any subset of spins $M\subseteq \Lambda$.
In particular, choosing $M=\Lambda$ we get a bound
\be
\label{even_more_trivial_bound}
\|V\|\le \sum_{A\subseteq \Lambda} \|V_A\| \le N \cdot \|V\|_1.
\ee
Applying Eq.~(\ref{even_more_trivial_bound}) we can bound the norm of $H_{\garb}$ as
\be
\label{tr_error}
\|H_{\garb}\| =\| e^T (H_0+\epsilon V)\, e^{-T} - H^{\la n\ra} \| \le N \sum_{j=n+1}^\infty |\epsilon|^j\cdot \|V^{(j-1)}\|_1.
\ee
Thus our task is to get an upper bound on the norm $\|V^{(j)}\|_1$.
To simplify notations we shall assume that the original perturbation $V$ has strength $\|V\|_1=1$.
\begin{lemma}
The operators $T_j$ and $V^{(j-1)}$ are $(j+1)$-local Hamiltonians.
There exists constants $\alpha,\beta>0$ such that
\label{lemma:v}
\be
\label{Gscaling}
\|V^{(j)}\|_1 \le  \alpha \left( \frac{n^2}{\Delta\cdot  \beta}\right)^{j}
\ee
for all $j\ge 0$.
\end{lemma}
Let us first explain how Theorem~\ref{thm:local} follows from Lemma~\ref{lemma:v}.
Define $\rho_c=\Delta \beta n^{-2}$. Assume for simplicity that $\epsilon>0$.
Substituting Eq.~(\ref{Gscaling})   into Eq.~(\ref{tr_error}) one gets
\be
\label{deltap}
\| H_{\garb}\| \le \alpha N \sum_{j=n+1}^\infty \epsilon^j \rho_c^{1-j} \le 2\alpha \beta \Delta N n^{-2} (\epsilon/\rho_c)^{n+1}
\ee
provided that $\epsilon<\rho_c/2$.
Using Corollary~\ref{corol:overall} of Lemma~\ref{lemma:new}
and the fact that $V^{(j-1)}$ is $(j+1)$-local Hamiltonian
we conclude that
the ground state of $H^{\la n\ra}$ belongs to the subspace $\calP_0$ whenever
\[
\sum_{j=1}^n 2^{3+j} \cdot |\epsilon|^j \cdot \|V^{(j-1)}\|_1 <\Delta.
\]
Using  the bound Eq.~(\ref{Gscaling})  we get
\[
\sum_{j=1}^n 2^{3+j} \cdot |\epsilon|^j \cdot \|V^{(j-1)}\|_1\le 16 \alpha \epsilon \sum_{j=0}^\infty (2\epsilon /\rho_c)^j
\le 32 \alpha \epsilon <\Delta
\]
provided that $\epsilon<\rho_c/4$ and $32\alpha \epsilon<\Delta$.
If we choose
\[
\epsilon_c=\min{\left[ \frac{\Delta \beta}{4n^2}, \frac{\Delta}{32 \alpha}\right]}
\]
then for all $\epsilon<\epsilon_c$ the ground state of $H^{\la n\ra}$ belongs to the subspace
$\calP_0$ and $\|H_{\garb}\|=O(N \epsilon^{n+1})$.
Theorem~\ref{thm:local} now follows from the fact that $H^{\la n\ra} + H_{\garb}$ has the same
spectrum as  $H_0+V$.
In the rest of the section we prove Lemma~\ref{lemma:v}.
\begin{proof}[\bf Proof of Lemma~\ref{lemma:v}]
Let us first analyze locality of  $T_j$ and $V^{(j-1)}$.
\begin{prop}
\label{prop:locality}
$T_j$ is a $(j+1)$-local anti-hermitian Hamiltonian.
$V^{(j-1)}$ is a $(j+1)$-local Hamiltonian.
\end{prop}
\begin{proof}
We shall use the fact that for any $a$-local Hamiltonian $X$ and $b$-local Hamiltonian $Y$
the commutator $[X,Y]$ is at most $(a+b-1)$-local.
By assumption, $V^{(0)}=V$ is a $2$-local Hamiltonian.
Using Eq.~(\ref{DFFR-method}) we conclude that $T_1$ is a $2$-local Hamiltonian (anti-hermitian).
Let us use induction to  prove that $T_j$ and $V^{(j-1)}$ are $(j+1)$-local.
Indeed, from Eq.~(\ref{V(j-1)T}) one can see that  $V^{(j-1)}$ is a sum of terms proportional to
$\hat{T}_{j_1} \cdots \hat{T}_{j_q}(H_0)$ with $j_1+\ldots+j_q=j$, $q\ge 2$, and terms proportional to
$\hat{T}_{j_1} \cdots \hat{T}_{j_q}(V)$ with  $j_1+\ldots+j_q=j-1$, $q\ge 1$. The terms of the first type
are at most $k'$-local with $k'=j_1+\ldots+j_q+1=j+1$.
The terms of the second type are
at most $k''$-local with $k''=2+j_1+\ldots+j_q=j+1$.
Thus $V^{(j-1)}$ is $j+1$-local. Using Eq.~(\ref{DFFR-method})
we conclude that $T_j$ is $j+1$-local.
\end{proof}
Recall that the truncated Hamiltonian $H^{\la n\ra}$ is computed using a {\em finite} number
of generators $T_1,\ldots,T_n$. Thus {\em any} generator $T_j$ is at most
$(n+1)$-local.

\begin{prop}
\label{prop:induction}
Let $W$ be an  $r$-local Hamiltonian.  Then for any $j\le n$ one has
\be
\label{observation1}
\| \hat{T}_j (W)\|_1 \le \frac{2(r+n+1)}{\Delta}\,   \| V^{(j-1)}\|_1\cdot \|W\|_1.
\ee
\end{prop}
\begin{proof}
Indeed, expand $T_j$ and $W$ as a sum of local operators,
\[
T_j=\sum_{M\subseteq \Lambda} T_{j,M}, \quad W=\sum_{A\subseteq \Lambda} W_A,
\]
where $|M|\le j+1\le n+1$, see Proposition~\ref{prop:locality}, and $|A|\le r$ since $W$ is $r$-local.
Here $T_j$ is related to $V^{(j-1)}$ by Eq.~(\ref{DFFR-method}), that is,
\[
V^{(j-1)}=\sum_{M\subseteq \Lambda} V^{(j-1)}_M, \quad T_{j,M}=\calL_M (V^{(j-1)}_M),
\quad \|T_{j,M}\| \le \frac1{\Delta}\, \|V^{(j-1)}_M\|.
\]
Let $u\in \Lambda$ be the spin that maximizes the norm $\| \hat{T}_j (W)\|_1$. Then
\bea
\| \hat{T}_j (W)\|_1 &\le & 2\sum_{M\ni u} \; \sum_{A\, : \, A\cap M\ne \emptyset} \| T_{j,M}\| \cdot \| W_A\| +
2\sum_{A\ni u} \; \sum_{M\, : \, M\cap A\ne \emptyset} \| T_{j,M}\| \cdot \| W_A\| \nn \\
&\le & \frac2{\Delta} \sum_{M\ni u} \; \sum_{A\, : \, A\cap M\ne \emptyset} \| V^{(j-1)}_{M}\| \cdot \| W_A\| +
\frac2{\Delta}  \sum_{A\ni u} \; \sum_{M\, : \, M\cap A\ne \emptyset} \| V^{(j-1)}_M\| \cdot \| W_A\| \nn \\
&\le &  \frac2{\Delta}  \|W\|_1\, \sum_{M\ni u} |M| \cdot \| V^{(j-1)}_{M}\| +\frac2{\Delta} \|V^{(j-1)}\|_1
\sum_{A\ni u} |A|\cdot \| W_A\| \nn \\
&\le &  \frac{2(r+n+1)}{\Delta}\,   \| V^{(j-1)}\|_1\cdot \|W\|_1.\nn
\eea
\end{proof}
Now we can apply Proposition~\ref{prop:induction} to bound the norm of any term in Eq.~(\ref{V(j-1)T}).
Note that $V^{(j-1)}$ is composed of two types of terms. We have terms proportional to
$A_{\bf j}=\hat{T}_{j_1} \cdots \hat{T}_{j_q}(H_0)$ with $1\le j_1,\ldots,j_q\le n$, $j_1+\ldots+j_q=j$, and terms proportional to
$B_{\bf j} =\hat{T}_{j_1} \cdots \hat{T}_{j_q}(V)$ with $1\le j_1,\ldots,j_q\le n$, $j_1+\ldots+j_q=j-1$.
Note that applying $q$ commutators $\hat{T}_{j_1}\cdots \hat{T}_{j_q}$ to $H_0$ or $V$ we can produce
at most $O(qn)$-local operator.
Applying Proposition~\ref{prop:induction} we get
\bea
\label{ABnorms}
\|A_{\bf j}\|_1 &\le &O(n\Delta^{-1})^q \, q^q \cdot \Delta \cdot \|V^{(j_1-1)}\|_1 \cdots \|V^{(j_q-1)}\|_1, \nn \\
\|B_{\bf j}\|_1 & \le &  O(n\Delta^{-1})^q \, q^q \cdot \|V^{(j_1-1)}\|_1 \cdots \|V^{(j_q-1)}\|_1
\eea
where we have  taken into account that $\|V\|_1=1$ and $\|H_0\|_1=O(\Delta)$.
Let us introduce an auxiliary quantity
\be
\label{chi}
\chi_j=\|V^{(j-1)}\|_1, \quad j\ge 1.
\ee
Using Eqs.~(\ref{V(j-1)T},\ref{ABnorms})
we can get an upper bound $\chi_j\le \mu_j$, where
\[
\mu_1=\chi_1= \|V\|_1=1
\]
 and
$\mu_j$, $j\ge 2$, are defined inductively from
\be
\label{mu_inductive}
\mu_j = \Delta \sum_{q=2}^j c^q
\sum_{\scriptscriptstyle j_1+ \ldots+ j_q=j}
\mu_{j_1} \cdots \mu_{j_q} +
\sum_{q=1}^{j-1} c^q
\sum_{ \scriptscriptstyle  j_1+ \ldots+ j_q=j-1}
\mu_{j_1} \cdots \mu_{j_q},
\ee
where $c\equiv O(n\Delta^{-1})$.
Here we have taken into account that $q^q/q!=O(1)^q$ and omitted the constraint
$j_1,\ldots,j_q\le n$ since we are interested in an upper bound on $\chi_j$ (note that adding
this constraint can only make $\mu_j$ to grow more slowly with $j$).
In order to get an upper bound on $\mu_j$
define an auxiliary Taylor series
\be
\label{mu_series}
\mu(z)=\sum_{j=1}^\infty \mu_j\, z^j.
\ee
Inspecting the recursive
equation Eq.~(\ref{mu_inductive}) one can easily convince oneself that
$\mu(z)$ has to obey the following equation
\[
\mu = \Delta \left( \frac1{1-c\mu} -1 -c\mu\right)  + z\left( \frac1{1-c\mu}-1\right) + z.
\]
Solving this equation one can easily find the inverse function:
\[
z(\mu)=\mu -\mu^2(c+\Delta c^2).
\]
Clearly $z(\mu)$ is analytic in the entire complex plane.
Solving the quadratic equation and choosing the branch for which $\mu(0)=0$ one gets
\[
\mu(z)=\frac1{2(c+\Delta c^2)} \, \left(1-\sqrt{1-4z(c+\Delta c^2)}\right).
\]
Clearly $\mu(z)$ is analytic in a disk $|z|<z_0$ where
\[
z_0=\frac1{4(c+\Delta c^2)} \ge  \frac{\beta\Delta}{n^2}
\]
for some constant $\beta>0$.
Here we have taken into account that $c=O(n\Delta^{-1})$.
Using the Cauchy integral formula for $\mu(z)$ one gets
\[
\mu_j =\frac1{2\pi i} \oint_{|z|=z_0/2} \frac{\mu(z)}{z^{j+1}} \, dz
\]
which yields
\[
\mu_j\le \alpha \left( \frac{\beta \Delta}{n^2}\right)^{1-j}
\]
for some constant $\alpha>0$.
The lemma is proved.
\end{proof}

\subsection{Equivalence between the local and global SW theories}
\label{subs:equiv}

In this section we prove Theorem~\ref{thm:global}.
The key idea of the proof is to construct a unitary operator $e^K$ acting on the low-energy subspace
$\calP_0$ that transforms the $n$-th order effective low-energy Hamiltonian
obtained using the local SW method to the one obtained using the global SW method
with an error  $O(N|\epsilon|^{n+1})$ in the operator norm.
We illustrate the relationship between the two low-energy theories on Fig.~\ref{fig:iso}.
Once the desired unitary transformation $e^K$ is constructed, Theorem~\ref{thm:global}
follows directly from Theorem~\ref{thm:local}.
{
\unitlength=1mm
\begin{figure}
\centerline{
\begin{picture}(90,40)(0,0)
\put(48,41){\oval(30,10)}
\put(40,40){$H_0+\epsilon V$}
\put(3,11){\oval(50,10)}
\put(-15,10){$P_0 \, e^S(H_0+\epsilon V)e^{-S}P_0$}
\put(95,11){\oval(50,10)}
\put(77,10){$P_0 \, e^T(H_0+\epsilon V)e^{-T} P_0$}
\put(38,36){\vector(-1,-1){20}}
\put(58,36){\vector(1,-1){20}}
\put(70,11){\vector(-1,0){42}}
\put(46,13){$e^K$}
\put(30,23){$e^S$}
\put(60,23){$e^T$}
\put(5,25){\mbox{Global SW}}
\put(71,25){\mbox{Local SW}}
\end{picture}}
\caption{Correspondence between the high-energy theory $H_0+\epsilon V$
and the effective low-energy theories obtained using the global and the local
Schrieffer-Wolff transformations. The unitary operator $e^K$ defined on the low-energy
subspace $\calP_0$ maps the two effective low-energy Hamiltonians to each other.}
\label{fig:iso}
\end{figure}
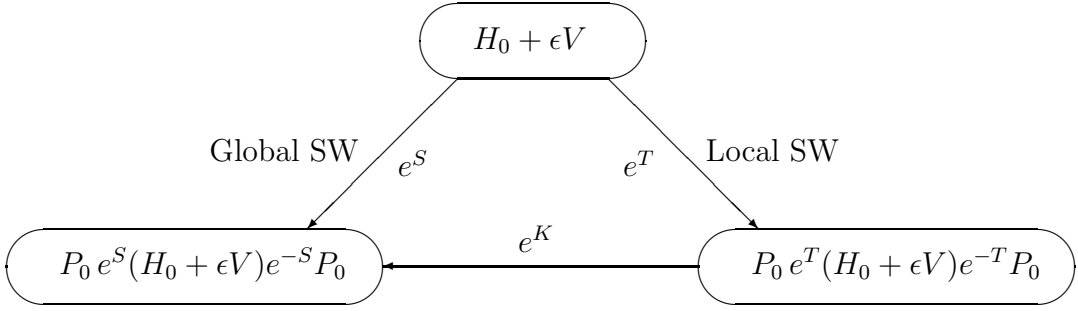
}
To introduce the notations let us first consider the limit $n\to \infty$.
Let
\be
\label{GSseries}
S=\sum_{j=1}^\infty S_j \, \epsilon^j \quad \mbox{and} \quad
T=\sum_{j=1}^\infty T_j \, \epsilon^j
\ee
be generators of the SW transformation calculated using the global rule ($S$) and
the local rule ($T$). Here we regard $S$ and $T$ as formal Taylor series.
Let
\[
H^{\gl}=e^S (H_0+\epsilon V) e^{-S} \quad \mbox{and} \quad
H^{\loc}=e^T (H_0+\epsilon V) e^{-T}
\]
be the corresponding  transformed Hamiltonians which we also consider as formal
Taylor series. Let $H^{\loc, \la n\ra }$ and $H^{\gl,\la n\ra }$
be the Hamiltonians obtained from $H^{\loc}$ and $H^{\gl}$
respectively by truncating the series at the order $n$.
\begin{theorem}
\label{thm:iso}
Let
\[
L=P_0\, H^{\loc,\la n\ra } P_0 \quad \mbox{and} \quad  M=P_0\, H^{\gl,\la n\ra } P_0
\]
 be the effective order-$n$
Hamiltonians on the low-energy subspace obtained by the local and the global SW methods.
There exists a constant $\beta$ depending only on $n$, $\Delta$, and the maximum degree of the interaction  graph
such that
\be
\| M- e^K  L e^{-K} \| \le \beta N |\epsilon|^{n+1} \quad \mbox{for all $|\epsilon|\le 1$}
\ee
for some unitary operator $e^{K}$ acting on the subspace $\calP_0$.
\end{theorem}

A natural way to define the desired unitary transformation $e^K$ is
\be
\label{eK}
e^K = e^S \, e^{-T}.
\ee
Indeed, one should expect that $e^{-T}$ maps the local low-energy to the high-energy theory
and $e^S$ maps the high-energy theory to the global low-energy theory.
The main difficulty in using Eq.~(\ref{eK}) is that (i) it defines a unitary transformation on the entire Hilbert space
$\calH$ while we expect the two theories to be close only on the low-energy subspace $\calP_0$;
(ii) the formal series $K$, $S$, and $T$ may be divergent.

The proof of Theorem~\ref{thm:iso} goes as follows.
We shall regard Eq.~(\ref{eK}) as a   definition of a  formal Taylor series
\[
K=\sum_{j=1}^\infty K_j \, \epsilon^j
\]
Here $K_j$ are some operators acting on the entire Hilbert space $\calH$.
We will show that operators $K_j$ are actually block-diagonal, so we can apply the
transformation $e^K$ on the subspace $\calP_0$ only, see Lemma~\ref{lemma:bd},
even if one uses a truncated version of $K$.
Next we shall make use of the fact  that the SW operators $e^S$ and $e^T$ have (weak) multiplicativity
properties, see Lemma~\ref{lemma:mult}.
It will allow us to prove multiplicativity for $e^K$ and thus additivity of $K$, see Eqs.~(\ref{Kadd}).
Next we shall exploit the additivity of $K$ to show that the series for $K$ obeys the linked
cluster condition analogous to Theorem~\ref{thm:lct}.
Finally, we shall construct the desired transformation mapping the
local theory to the global theory as $\exp{(K^{\la n\ra})}$, where
 $K^{\la n\ra}=\sum_{j=1}^n K_j \, \epsilon^j$. Lemma~\ref{lemma:KLM} shows that this transformation
 indeed does the right job.

\begin{lemma}
\label{lemma:bd}
Suppose $H_0$  has a positive spectral gap $\Delta$
between $\calP_0$ and $\calP_0^\perp$.  Then
all generators $K_j$ are block-diagonal.
\end{lemma}
\begin{proof}
By definition we have
\be
\label{loc2gl}
H^{\gl}=e^K\, H^{\loc} \, e^{-K}
\ee
and both operators $H^{\gl}$, $H^{\loc}$ are block-diagonal.
It is assumed here that $H^{\gl}$ and $H^{\loc}$ are formal Taylor series in $\epsilon$
Let us multiply Eq.~(\ref{loc2gl}) by $P_0$ and $Q_0$ on the left and on the right respectively.
Taking into account that $H^{\gl}$ and $H^{\loc}$ are block-diagonal, and using formal Taylor
series $H^{\loc}=\sum_{j=0}^\infty H^{\loc}_j \, \epsilon^j$ with $H^{\loc}_0=H_0$ one obtains
\be
\label{loc2gl'}
P_0 [K_j,H_0] Q_0 + P_0 f(K_1,\ldots,K_{j-1},H^{\loc}_1,\ldots,H^{loc}_j) Q_0=0 \quad \mbox{for all $j\ge 1$},
\ee
where $f(K_1,\ldots,K_{j-1},H^{\loc}_1,\ldots,H^{loc}_j)$ is some polynomial obtained from the
expansion of the exponent in Eq.~(\ref{loc2gl}).
Let us use induction in $j$ to show that $K_j$ is block-diagonal.
Indeed, suppose we have already proved that $K_1,\ldots,K_{j-1}$ are block-diagonal.
Then the term
\[
P_0 f(K_1,\ldots,K_{j-1},H^{\loc}_1,\ldots,H^{loc}_j) Q_0
\]
 in Eq.~(\ref{loc2gl'}) vanishes and we
can rewrite Eq.~(\ref{loc2gl'}) as
\be
\label{loc2gl''}
(P_0 K_j Q_0)  (Q_0 H_0Q_0)=(P_0H_0P_0)(P_0K_j Q_0).
\ee
If $|\psi\ra \in Q_0$ is any eigenvector of $Q_0H_0 Q_0$ with an eigenvalue $\lambda$,
then Eq.~(\ref{loc2gl''}) implies that a state $(P_0K_jQ_0)\,|\psi\ra$
is an eigenvector of $P_0H_0P_0$ with
the same eigenvalue $\lambda$. On the other hand,
$(P_0K_jQ_0)\,|\psi\ra\in \calP_0$.
 It contradicts to the spectral gap assumption
unless $(P_0K_jQ_0)\, |\psi\ra=0$. Since eigenvectors of $Q_0H_0Q_0$ span the entire subspace $Q_0$,
we conclude that $P_0K_j Q_0=0$. Since $K$ is anti-hermitian, it implies $Q_0K_j P_0=0$, that is,
$K_j$ is block-diagonal.
\end{proof}

Let us now prove that the restriction of $K$ onto the $\calP_0$ subspace obeys the additivity property:
if one partitions the lattice $\Lambda$ into two disjoint parts, $\Lambda=AB$, such that $V=V_A+V_B$
then $K=K_A+K_B$ (only on the subspace $\calP_0$).
Note that the generator $T$ constructed using the local SW method  is additive by construction, so
we have $T=T_A+T_B$ on the entire Hilbert space (not only on the subspace $\calP_0$).
On the other hand,  weak multiplicativity of the direct rotation, see Lemma~\ref{lemma:mult},
implies
\be
\label{Sadd}
P_0 \, e^{S}=P_0 \, e^{S_A+ S_B}.
\ee
Since $e^K =e^S e^{-T}$, we conclude that
the restriction of $K$ onto the subspace $\calP_0$ is additive, that is,
\be
\label{Kadd}
e^{-K}\, P_0 = e^{-K_A - K_B}\, P_0
\ee
By abuse of notations let us identify  $K$, $K_A$, $K_B$ with  their restrictions of the on the subspace $\calP_0$.
Then we have $e^{-K}=e^{-K_A} e^{-K_B} = e^{-K_A -K_B}$.
Taking the logarithm we obtain the desired additivity property,
$K=K_A+K_B$.
It implies  that the series for $K$ obeys the linked cluster condition, see Theorem~\ref{thm:lct}.
Using the same notations as in the statement of Theorem~\ref{thm:lct}
we conclude that
\[
K_j=\sum_{C\subseteq \calE} K_{j,C},
\]
where $K_{j,C}$ acts only on $\Lambda(C)$ and $K_{j,C}=0$ unless
$|C|\le j$ and $C$ is a connected subset of edges.
Define an interaction strength of $K$ as
\be
\label{strength1}
J_K=\max_{u\in \calL} \max_{j=1,\ldots,n} \;  \sum_{C\ni u} \|K_{j,C}\|.
\ee
We claim that
\be
\label{strength2}
J_K=O(1),
\ee
that is, $J_K$ has an upper bound independent of $N$. Indeed,
assuming that the interaction graph has degree $d=O(1)$, the number of linked clusters of size $\le n$
containing a given vertex $u$ can be upper bounded by $d^{O(n)}=O(1)$.
Additivity of $K$ implies that we can
compute $K_{j,C}$ by turning off all interactions $V_{u,v}$, $(u,v)\notin C$.
It leaves us with a subgraph with at most $N'=2|C|\le 2j\le 2n=O(1)$ sites.
Since the number of spins is $O(1)$, any operator $K_{j,C}$ is a constant-degree polynomial
 with constant-bounded coefficients, so that $\|K_{j,C}\|=O(1)$.
Theorem~\ref{thm:iso} now follows from the following lemma.
\begin{lemma}
\label{lemma:KLM}
Let $K,L,M$ be formal operator-valued Taylor series acting on the subspace $\calP_0$ such that
$K$ is anti-hermitian, $K(0)=0$, and
\be
\label{cond1}
M=e^K L e^{-K}.
\ee
For any integer $n=O(1)$  define truncated series
\be
\tilde{K}=\sum_{j=1}^n K_j\, \epsilon^j, \quad \tilde{L}=\sum_{j=0}^n L_j\, \epsilon^j, \quad \tilde{M}=\sum_{j=0}^n M_j\, \epsilon^j
\ee
and a truncation error
\be
\delta=\| \tilde{M} -e^{\tilde{K}} \tilde{L} \, e^{-\tilde{K}}\|.
\ee
Suppose $K,L,M$ obey the linked cluster property and have constant-bounded interaction strength,
$J_K,J_L,J_M=O(1)$. Then
\be
\delta\le O(1) N |\epsilon|^{n+1} \quad \mbox{for any $|\epsilon|\le 1$}.
\ee
\end{lemma}
\begin{proof}
Introduce a truncated version of a transformed operator $e^{\tilde{K}} \tilde{L} \, e^{-\tilde{K}}$, namely,
\[
\Theta=\tilde{L} + \sum_{j=1}^n \frac1{j!} \, [\tilde{K},\cdot]^j \tilde{L}.
\]
The assumptions that $K$ obeys the linked cluster property and $K(0)=0$
imply that  $i\tilde{K}$ can be regarded as a local Hamiltonian with interaction strength
upper bounded by $J_K \sum_{j=1}^n |\epsilon|^j =O(|\epsilon|)$.
Similarly, $\tilde{L}$ can be regarded as a local Hamiltonian with interaction strength $O(1)$.
Using  Lemmas~1,2 from~\cite{gadget_prl} one gets an approximation
\be
\label{approx1}
\| e^{\tilde{K}} \tilde{L} \, e^{-\tilde{K}} - \Theta\| \le \frac1{(n+1)!} \, \| \,
 [\tilde{K},\cdot]^{n+1} \tilde{L}\, \| \le
O(1) N |\epsilon|^{n+1}.
\ee
Therefore it suffices to show that
\be
\label{approx2}
\| \tilde{M} - \Theta \| \le O(1) N |\epsilon|^{n+1}.
\ee
Expanding the multiple commutators in $\Theta$ we get a linear combination of terms
\[
\frac1{j!}\,
\epsilon^{q_1+q_2+\ldots+q_j+q} \, \hat{K}_{q_1} \hat{K}_{q_2} \cdots \hat{K}_{q_j}(L_q), \quad 1\le q_1,\ldots,q_j\le n, \quad 1\le j\le n, \quad 1\le q\le n.
\]
Let us call a term as above {\it good} if the total power of $\epsilon$ is at most $n$, that is,
$q_1+q_2+\ldots+q_j+q\le n$. Otherwise let us call such a term {\it bad}. Using the condition Eq.~(\ref{cond1})
one can easily check that  the sum of all good terms in $\Theta$ equals $\tilde{M}$. Thus we get a bound
\be
\label{approx3}
\| \tilde{M} - \Theta \| \le |\epsilon|^{n+1} \, \#\{ \mbox{bad terms}\} \, \max{\| \hat{K}_{q_1} \hat{K}_{q_2} \cdots \hat{K}_{q_j}(L_q)\|}.
\ee
Clearly for a constant $n$ the total number of terms in $\Theta$ is $O(1)$ and so the number of bad terms is $O(1)$.
Applying again Lemma~2 from~\cite{gadget_prl} and taking into account that
the Taylor coefficients $K_{q_i}$ and $L_q$ are local Hamiltonians with interaction strength $O(1)$ we get a bound
\[
\max{\| \hat{K}_{q_1} \hat{K}_{q_2} \cdots \hat{K}_{q_j}(L_q)\|}\le O(1) N.
\]
Therefore $\| \tilde{M} - \Theta\| \le O(1) N |\epsilon|^{n+1}$.
\end{proof}

\section*{Acknowledgments}
We thank Barbara Terhal for useful discussions.
SB would like to thank RWTH Aachen University
and University of Basel for hospitality  during several stages of this work.
SB was partially  supported by the
DARPA QuEST program under contract number HR0011-09-C-0047.
DL was partially supported by the Swiss NSF, NCCR Nanoscience, NCCR QSIT,
and DARPA QuEST.

\end{document}